\documentclass[a4paper,twoside]{article}
\usepackage{amsmath,amssymb,amsthm}
\usepackage{amscd}
\usepackage[usenames,dvipsnames]{color}
\usepackage{pstricks}
\usepackage{graphicx}
\usepackage[all]{xy}
\usepackage{tabularx}

\newtheorem{prop}{Proposition}[section]

\newtheorem{defi}{Definition}[section]
\newtheorem{lemm}{Lemma}[section]
\newtheorem{thm}{Theorem}[section]
\newtheorem{coro}{Corollary}[section]
\newtheorem{ex}{Example}[section]

\begin{document}
\date{}
%
%
\title{\textbf{The extension of distributions on manifolds, a microlocal approach.}}
\author{
Nguyen Viet Dang\\
Laboratoire Paul Painlev\'e (U.M.R. CNRS 8524)\\
UFR de Math\'ematiques\\
Universit\'e de Lille 1\\
59 655 Villeneuve d'Ascq C\'edex France.
}%
\maketitle

\begin{abstract}
Let $M$ be a smooth manifold, $I\subset M$
a closed embedded submanifold of $M$ and
$U$ an open subset of $M$.
In this paper, we find conditions
using a geometric notion of scaling
for $t\in \mathcal{D}^\prime(U\setminus I)$
to admit an extension in $\mathcal{D}^\prime(U)$.
We give microlocal conditions
on $t$ which allow
to control the wave front set
of the extension generalizing a previous result of
Brunetti--Fredenhagen.
Furthermore, we show that there is a subspace
of extendible 
distributions for which 
the wave front of the extension
is minimal which has 
applications for the
renormalization of quantum field theory
on curved spacetimes.
\end{abstract}


\tableofcontents

\section*{Introduction.}

From the early days of quantum field theory, it has been known~\cite{Bogoliubov,Dyson,Dyson-49} that
QFT calculations are plagued
with infinities arising
from the integration of divergent 
Feynman amplitudes in
momentum space.
The recipe devised 
to subtract these divergences 
is called the renormalization 
algorithm~\cite{Collins}. 
When one generalizes 
QFT to curved Lorentzian spacetimes~\cite{Birrel,Fulling,Wald}, a simple
observation is that
both the conventional 
axiomatic approach to quantum field theory 
following Wightman's axioms~\cite{StWi:pct} or the usual textbook 
approach based on the representation of Feynman amplitudes
in momentum space, completely break down 
for the
obvious reasons that there is no Fourier transform
on curved spacetime and the spacetime is no longer Lorentz invariant.\\
This motivates to look at the renormalization problem
of Feynman amplitudes from the point of view
of the position space and 
this problem was solved 
in the seminal work
of Brunetti and Fredenhagen 
\cite{Brunetti2}.  
The starting point of
\cite{Brunetti2}
was to follow 
one of the very 
first
approach
to QFT 
due to
Stueckelberg and his collaborators
(D. Rivier, T.A. Green, A. Petermann), 
which is based 
on the concept of causality.

The ideas of Stueckelberg 
were first 
understood
and developed 
by Bogoliubov and his school (\cite{Bogoliubov}) 
and then 
by Epstein-Glaser (\cite{Epstein}, \cite{EGS}) 
(on flat spacetime).
In these approaches, one works
directly
in position space 
and the renormalization 
is formulated as a problem
of extension of distributions.
Somehow, this point of view based 
on the S-matrix formulation of QFT 
was almost
completely forgotten
by people working on QFT 
at the exception of
some works~\cite{Kay2,Kay-97,Kay1,Popineau,Stora02}.
However, in 1996, 
a student of Wightman,
M.\,Radzikowsky 
revived the subject. 
In his thesis~\cite{Junker,Junker-96,Radzikowski-92-PhD}, 
he used microlocal analysis
for the first time 
in this context
and introduced the concept
of 
\emph{microlocal spectrum condition},
a condition on the wave front set
of the distributional
two-point function
which represents the quantum states 
of positive energy 
(named Hadamard states)
on curved spacetimes~\cite{gerard2014constructionCauchy,gerard2014construction,gerard2014hadamard,wrochna2013singularities}.
In 2000, in a breakthrough paper, 
Brunetti and Fredenhagen
were able to generalize 
the Epstein-Glaser theory
on curved spacetimes 
by relying on the fundamental contribution of 
Radzikowski. 
These results were further extended
in some exciting recent works~\cite{derezinski2013mathematics,HollandsOPE,Hollands, Hollands2, Hollands4,Hollands5,nikolov2014renormalization} where
the formalism of algebraic QFT now includes the treatment of gauge theories like Yang-Mills fields~\cite{Hollands-08,wrochna2014classical},
and also incorporates the Batalin Vilkovisky formalism \cite{Fredenhagen-11,Fredenhagen-13} in order
to perturbatively quantize gravitation~\cite{Brunetti-13,Brunetti-13-QG,Rejzner-PhD}.

All the above works rely on a formalism for renormalization theory
which consists in a recursive procedure of extension of products of distributions
representing Feynman amplitudes on configuration space.
More precisely, if we denote by $\Delta_F$ the Feynman propagator
which is a fundamental solution of the Klein Gordon operator
$\square+m^2$ with a specific wave front set, then a Feynman amplitude
will be a product of the form $\prod_{1<i\leqslant j<n}\Delta_F^{n_{ij}}(x_i,x_j) $, this product
of distribution is well defined
on the configuration space $M^n$ minus all diagonals $x_i=x_j$ since all the wave front sets
are transverse and renormalization consists in extending the above product
on the whole configuration space $M^n$ in a way which is compatible with the physical
axiom of causality.
The central technical ingredient of the recursive procedure is to control
the wave front set and the \emph{microlocal scaling degree} of the renormalized products 
in such a way that
we can construct all Feynman amplitudes on all
configuration spaces $M^n$ by induction on $n\in\mathbb{N}$.
In the present paper, which is an outgrowth of \cite{Dangthese},
our goal is to build
some scale spaces of distributions on manifolds, 
study their intrinsic property then  
discuss 
the operations of extension and
renormalization of products relying
on recent works 
on the functional analytic properties of the space
$\mathcal{D}^\prime_\Gamma$ of
distributions with given wave front set
\cite{Viet-wf2,dabrowski2013functional,dabrowski2014functional}. 
An interesting perspective for future 
investigations is to study how our
renormalization preserves or breaks symmetries
of distributions
in the spirit of \cite{bahns2012shell}.
\paragraph*{Acknowledgements.}
We would like to thank Christian Brouder and Fr\'ed\'eric H\'elein for
their collaboration, many enlightening discussions and their support 
which results in the present paper. 
We also thank Yoann Dabrowski for his many advices concerning
the functional analytic aspects of the space $\mathcal{D}^\prime_\Gamma$
and his deep influence on the way the author thinks about the extension problem.
Finally I would like to thank the Labex CEMPI and the Laboratoire Paul Painlev\'e for
excellent working conditions.

The following section is 
a detailed overview of our
results and can
be read independently from
the rest.
\subsection*{Main results of our paper.}
In our paper, we investigate 
the following
problem which has simple formulation:
we are given a manifold $M$ and a closed submanifold
$I\subset M$. We have a distribution $t$
defined on $M\setminus I$ and we would like to
find under what reasonable conditions
on $t$, 
\begin{enumerate}
\item we can construct an extension
$\overline{t}$ of $t$ defined on the
whole manifold $M$,
\item we can control the wave front set
of the extension.
\end{enumerate}

The first problem has been addressed
in greater generality in \cite{dang2014extension} 
where
we found necessary and sufficient
conditions for a distribution $t\in\mathcal{D}^\prime(M\setminus I)$
where $I$ 
is a \textbf{closed subset} of $M$,
to be extendible.
However, our method which uses distance functions,
is only adapted to Euclidean QFT. Actually, for 
QFT on curved Lorentzian spacetimes, 
it is crucial to find 
estimates
on the wave front set of the
extension. This is what we do in the present paper
which is focused entirely 
on the microlocal approach.

In general, the extension
problem
has no positive answer for 
a generic
distribution $t$
in $\mathcal{D}^\prime(M\setminus I)$
unless $t$
has moderate growth
when we approach the singular 
submanifold $I$.
In the work of Yves Meyer \cite{Meyer-98} where
the manifold $M$ is flat space $\mathbb{R}^n$
and $I=\{0\}$, the distributions
having this property
are called \emph{weakly homogeneous
distributions}.
A first difficulty
is to extend the definition
of Meyer to the case of manifolds.
In order to generalize the definition of scaling 
to measure the growth of distributions,
we introduce a 
class of 
vector fields
called Euler vector 
fields associated 
to the submanifold
$I$:
\begin{defi}\label{Eulerdef}
Let $M$ be a smooth manifold, $I$ a submanifold of $M$ and $U$ some open subset of $M$.
Set
$\mathcal{I}(U)$ to
be the ideal
of functions vanishing on $I$
and $\mathcal{I}^k(U)$ its $k$-th power.
A vector field $\rho$ locally defined on $U$ 
is called Euler if 
\begin{equation}
\forall f\in\mathcal{I}(U),  \rho f-f\in\mathcal{I}^2(U). 
\end{equation} 
\end{defi}
The above 
definition
is obviously intrinsic.
In particular, when $M=\mathbb{R}^d,I=\{0\}$
then $\rho=h^j\frac{d}{dh^j}$ is Euler.

In section \ref{scalingmfdsection}, the main properties 
of 
Euler vector fields are studied. 
They
satisfy a property
of diffeomorphism
invariance 
\begin{thm}\label{eulerdiffinvariance}
Let $M,M^\prime$ be two smooth manifolds,
$I\subset M,I^\prime\subset M^\prime$
smooth embedded submanifolds and $\Phi:=U\longmapsto U^\prime$ 
a local diffeomorphism
such that $\Phi(I\cap U)= I^\prime\cap U^\prime$.
Then
for any
Euler vector field
$\rho$ defined on $U$,
the pushforward
$\Phi_*\rho$
is Euler.
\end{thm}
And that
the flow generated by Euler
vector fields
are always locally
conjugate.
\begin{thm}\label{eulerlocconjugate}
Let $\rho_1,\rho_2$ be two Euler vector fields
defined in some neighborhood of $p\in I$.
Then there is some germ of diffeomorphism 
$\Phi$ at $p$ such that
$\rho_1=  \Phi_*\rho_2 $.
\end{thm}

In the sequel, once we are
given an Euler
vector field
$\rho$, let
$(e^{t\rho})_{t}$ be
the
one parameter
group of diffeomorphisms
generated
by $\rho$ then
we
will be interested by the 
one parameter group of scaling
flows
$(e^{\log\lambda\rho})_{\lambda\in(0,1]}$ 
and
the open subsets
$U$ which are 
stable by the flow 
$(e^{\log\lambda\rho})_{\lambda\in(0,1]}$
are called \emph{$\rho$
convex}.

For every
manifold $M$ and $I\subset M$
a closed embedded 
submanifold,
we construct in section \ref{Esiuconstructionsection} 
a collection 
of spaces
$(E_{s,I}(U))_U$, indexed by open subsets of $M$, of
\emph{weakly homogeneous
distributions
of degree} $s$ 
where $E_{s,I}(U)\subset \mathcal{D}^\prime(U)$,
with the following properties:
\begin{enumerate}

\item $E_{s,I}$ satisfies a restriction
property, if $V\subset U$ then
the restriction of $E_{s,I}(U)$ on $V$
is $E_{s,I}(V)$ and satisfies the following gluing property,
if $\cup_i V_i$ is an open cover of $U$ s.t. 
$\cup_i int\left(\overline{V_i}\right)$ is a neighborhood of $U$, then
for $t\in\mathcal{D}^\prime(\cup_i V_i)$,
$t\in E_{s,I}(V_i),\forall i\implies t\in E_{s,I}(U)$.

\item $E_{s,I}$ has the
important property of diffeomorphism
invariance:
\begin{thm}\label{Esdiffinvariance}
Let $M,M^\prime$ be two smooth manifolds,
$I\subset M,I^\prime\subset M^\prime$
smooth embedded submanifolds and $\Phi:=U\longmapsto U^\prime$ 
a local diffeomorphism
such that $\Phi(I\cap U)= I^\prime\cap U^\prime$.
Then $\Phi^*E_{s,I}(U^\prime)=E_{s,I}(U)$.
\end{thm}
\item The following proposition
gives a
concrete characterization
of elements in $E_{s,I}(U)$ for arbitrary open sets $U$ which 
could be used as a definition of $E_{s,I}(U)$:
\begin{prop}\label{caracterisationEs}
$t$ belongs to the
local space $E_{s,I}(U)$ if and only if
for all $p\in  int\left(\overline{U}\right)\cap I$, there is
some open chart $\psi:V_p\subset int\left(\overline{U}\right)\mapsto \mathbb{R}^{n+d}$, 
$\psi(I)\subset \mathbb{R}^n\times\{0\}$
where $\lambda^{-s}(\psi_*t)(x,\lambda h) $
is bounded in $\mathcal{D}^\prime(\psi(V_p\cap U))$.
\end{prop}
However, the property of diffeomorphism invariance
imply that
$E_{s,I}$ \textbf{does not depend on the choice
of Euler vector fields.}  
In particular
in the flat case 
where
$M=\mathbb{R}^{n+d}$
with coordinates $(x,h)=(x^i,h^j)_{1\leqslant i\leqslant n,1\leqslant j\leqslant d}$
and
$I=\{h=0\}$,
$t\in E_{s,I}(\mathbb{R}^{n+d})$
if $\left(\lambda^{-s}t(x,\lambda h)\right)_{\lambda\in(0,1]}$
is a bounded family
of distributions in $\mathcal{D}^\prime(V)$
where $\overline{V}$ is some neighborhood of $I$.
\item The 
collection
$(E_{s,I})_{s\in\mathbb{R}}$
is filtered, $s^\prime\geqslant s \implies 
E_{s,I}\subset E_{s^\prime,I}$ and
$E_{s,I}$ satisfies
an extension property (section \ref{extthmsection}):
\begin{thm}\label{extthmEsintro1}
Let $U\subset M$ be some open set. If
$t\in E_{s,I}(U\setminus I)$ then $t$
is extendible. Conversely,
if $t\in\mathcal{D}^\prime(M)$ then
for any bounded open set $U\subset M$,
$t\in E_{s,I}(U)$
for some $s\in\mathbb{R}$. 
\end{thm}
Moreover,
\begin{thm}\label{extthmEsintro2}
For all $s\in\mathbb{R}$, there is a linear map
\begin{eqnarray*}
t\in E_{s,I}(U\setminus I)\longmapsto \overline{t}\in  E_{s^\prime,I}(U)
\end{eqnarray*}
where $s^\prime = s$ if 
$s+d\notin -\mathbb{N}$ and $s^\prime< s$ otherwise.
\end{thm}
\end{enumerate}
Using diffeomorphism invariance and locality
of $E_{s,I}$,
the proof of the above property
is a consequence
of the microlocal
extension Theorem \ref{microlocextensionthm}
proved in 
the flat case.
The space $E_{s,I}$
only takes into account
the growth of distributions
along the submanifold $I$
which is not enough for many
applications, in particular
in quantum field theory where
we need to know the wave front set
of the extension $\overline{t}$
since we must 
\emph{multiply
distributions to define
Feynman amplitudes}. Therefore,
we need to refine
the definition
of weakly homogeneous
distributions, let
us introduce the necessary
definitions to
state our theorem.
We work in $\mathbb{R}^{n+d}$ with coordinates $(x,h)=(x^i,h^j)_{1\leqslant i\leqslant n,1\leqslant j\leqslant d}$, $I=\mathbb{R}^n\times\{0\}$ is the linear subspace $\{h=0\}$.
We assume $U$ to be of the form
$U_1\times U_2$ where $U_1$ (resp $U_2$) is an open subset of $\mathbb{R}^n$
(resp $\mathbb{R}^d$) s.t. $\lambda U_2\subset U_2,\forall\lambda\in(0,1]$.

We denote by $(x,h;\xi,\eta)$ 
the coordinates 
in cotangent space $T^* U$, 
where $\xi$ (resp $\eta$) is dual to $x$ (resp $h$).
$T^\bullet U$ denotes the cotangent $T^*U$ 
minus the zero section $\underline{0}$. 
If $U$ is \emph{convex}, 
then a set $\Gamma\subset T^\bullet U$ 
is \emph{stable by scaling}
if 
\begin{eqnarray}
\forall\lambda\in(0,1], \left( \{(x,\lambda^{-1}h;\xi,\lambda\eta) 
; (x,h;\xi,\eta)\in\Gamma   
\}\cap T^\bullet U\right) \subset \Gamma.
\end{eqnarray}
For $\Gamma$
a closed conic set
in $T^\bullet U$, $\mathcal{D}^\prime_\Gamma(U)$
is the 
space of distributions
in $\mathcal{D}^\prime(U)$
with wave front set
in $\Gamma$.
For $U$ a convex set
and $\Gamma\subset T^\bullet U$ 
a closed conic set
stable by scaling,
we denote
by 
$E_s(\mathcal{D}^\prime_\Gamma(U))$
the locally convex space 
of
\emph{weakly homogeneous distributions
of degree $s$ in $\mathcal{D}^\prime_\Gamma(U)$} 
defined
as follows:
$t\in E_s(\mathcal{D}^\prime_\Gamma(U))$
if $(\lambda^{-s}t(x,\lambda h))_{\lambda\in(0,1]}$
is a bounded family
of distributions in $\mathcal{D}^\prime_\Gamma(U)$.

We denote by $N^*\left(I\right)$ the conormal
bundle of $I$.
The central result of our paper is 
a general extension theorem 
(subsection \ref{genexttheoremssubsection}) for distributions
in flat space with control of the wave front set:
\begin{thm}\label{microlocextensionthm}
Let $U\subset \mathbb{R}^{n+d}$ be of the form
$U_1\times U_2$ where $U_1$ (resp $U_2$) is an open subset of $\mathbb{R}^n$
(resp $\mathbb{R}^d$) s.t. $\lambda U_2\subset U_2,\forall\lambda\in(0,1]$
and $\Gamma$ some closed conic set
in $T^\bullet U$. 
Set $\Xi=\{(x,0;\xi,\eta)|  (x,h;\xi,0)\in \Gamma\}\subset T^*_IU$.
For all $s\in\mathbb{R}$ there exists a linear, bounded map 
$t\in E_s(\mathcal{D}^\prime_\Gamma(U\setminus I))\longmapsto
\overline{t}\in  E_{s^\prime}(\mathcal{D}^\prime_{\Gamma\cup\Xi\cup N^*(I)}(U))$, 
where $s^\prime = s$ if 
$s+d\notin -\mathbb{N}$ and $s^\prime< s$ otherwise.
\end{thm}
An immediate
corollary of the above theorem
is the bound
$WF(\overline{t}) \subset\left( WF(t)\cup\Xi\cup N^*(I)\right)$
on the wave front of the extension. The central
ingredients of the proof are: 
a partition 
of unity formula which is a continuous analog
of the Littlewood--Paley decomposition used by Meyer \cite{Meyer-98},
to consider $(\lambda,x,h)\mapsto \lambda^{-s}t(x,\lambda h)$
as a distribution on the extended space $\mathbb{R}\times\mathbb{R}^{n+d}$
and a new integral formula
for the extension which reduces the bounds 
on the wave front set as applications of the Theorems in \cite{Viet-wf2}.

A particular case of the above
theorem was proved 
by Brunetti and Fredenhagen \cite{Brunetti2}
when 
the closure 
$\overline{\Gamma}$
of $\Gamma$
over $I$ 
is contained in
$N^*(I)$.
In that case, 
one can choose an 
extension $\overline{t}$
such that
$WF(\overline{t})\subset WF(t)\cup N^*(I)$ and
$\overline{t}\in  E_{s^\prime}(\mathcal{D}^\prime_{\Gamma\cup N^*(I)}(U))$.
The important condition
$\left(
\overline{\Gamma}\cap T_I^\bullet M\right) 
\subset N^*(I)$
called \emph{conormal landing
condition}
is \emph{intrinsic}
and
generalizes in a straightforward way to manifolds. 
It is a kind of microlocal regularity
condition and ensures that the wave front set 
of the extension is minimal.

Motivated by this intrinsic geometric
condition and the result of
Theorem \ref{microlocextensionthm},
we construct in section \ref{Esmicrolocsection} 
a subspace $E_{s,N^*(I)}\subset E_{s,I}$
which satisfies the following properties:
\begin{enumerate}
\item $E_{s,N^*(I)}$ satisfies the 
same restriction and gluing properties as $E_{s,I}$
\item $E_{s,N^*(I)}$ has the
important property of diffeomorphism
invariance:
\begin{thm}
Let $M,M^\prime$ be two smooth manifolds,
$I\subset M,I^\prime\subset M^\prime$
smooth embedded submanifolds and $\Phi:=U\longmapsto U^\prime$ 
a local diffeomorphism
such that $\Phi(I\cap U)= I^\prime\cap U^\prime$.
Then $\Phi^*E_{s,N^*(I)}(U^\prime)=E_{s,N^*(I)}(U)$.
\end{thm}
A consequence of the above diffeomorphism
invariance is that the definition
of $E_{s,N^*(I)}$ \textbf{does not depend on the choice
of Euler vector fields.}  
\item The 
collection
of spaces
$(E_{s,N^*(I)})_{s\in\mathbb{R}}$
is filtered, $s^\prime\geqslant s \implies 
E_{s,N^*(I)}\subset E_{s^\prime,N^*(I)}$
\end{enumerate}
The subspace $E_{s,N^*(I)}$ satisfies 
an extension theorem (section \ref{extthmsectionmicroloc})
\begin{thm}\label{microlocextensionfinal}
Let $U\subset M$ be some
open neighborhood of $I$. 
If $t\in E_{s,N^*(I)}(U\setminus I)$ 
then there exists an extension $\overline{t}$
with $WF(\overline{t})\subset WF(t)\cup N^*(I)$ and
$\overline{t}\in  E_{s^\prime,N^*(I)}(U)$, 
where $s^\prime = s$ if 
$s+d\notin -\mathbb{N}$ and $s^\prime< s$ otherwise.
\end{thm}
The main interest of this subspace is that
the wave front set $WF(\overline{t})$ of the
extension $\overline{t}$ 
is \emph{minimal}
in the sense we only add
the conormal $N^*(I)$ to
$WF(t)$. 
Then in section \ref{renormprodsection}, we present an application
of the above theorem to 
\textbf{renormalize products of distributions}, we denote by 
$E^\rho_s(\mathcal{D}^\prime_\Gamma(U))$ the space of distributions
$t$ s.t. the family $\left(\lambda^{-s}e^{\log\lambda\rho*}t\right)_{\lambda\in(0,1]}$ is
bounded in $\mathcal{D}^\prime_\Gamma(U)$ for some $\rho$-convex set $U$ and some cone $\Gamma$ stable by scaling:
\begin{thm}\label{renormprodthm}
Let $\rho$ be some Euler vector field, $U$ some neighborhood of $I$, $\left(\Gamma_1,\Gamma_2\right)$ two cones in $T^\bullet \left(U\setminus I\right)$
which satisfy the conormal landing condition and $\Gamma_1\cap -\Gamma_2=\emptyset$.
Set $\Gamma=\left(\Gamma_1+\Gamma_2\right)\cup\Gamma_1\cup\Gamma_2$.
If $\Gamma_1+\Gamma_2$ satisfies the conormal landing condition then
there exists a bilinear map $\mathcal{R}$ satisfying the following properties: 
\begin{itemize}
\item $\mathcal{R}:(u_1,u_2)\in E^\rho_{s_1}\left(\mathcal{D}^\prime_{\Gamma_1}(U\setminus I) \right)\times E^\rho_{s_2}\left(\mathcal{D}^\prime_{\Gamma_2}(U\setminus I) \right) \mapsto \mathcal{R}(u_1u_2)\in E_{s,N^*(I)}\left(U\right), \forall s<s_1+s_2$
\item $\mathcal{R}(u_1u_2)=u_1u_2\text{ on }U\setminus I$
\item $\mathcal{R}(u_1u_2)\in  \mathcal{D}^\prime_{\Gamma\cup N^*(I)}(U).$
\end{itemize}
\end{thm}
The above actually means that
$\mathcal{R}(u_1u_2)\in \mathcal{D}^\prime_{\Gamma\cup N^*(I)}(U)$ 
is a \textbf{distributional extension}
of the H\"ormander product $u_1u_2\in\mathcal{D}^\prime_\Gamma(U\setminus I)$.

In section \ref{renormambiguities}, 
we study the renormalization ambiguities which aim to classify the various
extensions we constructed.

\section{Scaling on manifolds.}
\label{scalingmfdsection}
\paragraph{Introduction.}
To solve the extension problem
for distributions
on manifolds,
we define in \ref{Eulerdef} a class of 
Euler vector fields 
which scale transversally
to a given fixed submanifold $I\subset M$.
In this section, we discuss the most important 
properties of this class of vector fields and 
their flows.
\begin{ex}\label{fundexEuler}
If $M=\mathbb{R}^{n+d}$ and $I$ is the vector subspace
which is the zero locus $\{h^j=0\}$ of the collection of
coordinate functions $(h^j)_j$, then 
$h^j\partial_{h^j}$ is Euler. Indeed, by application of Hadamard's lemma, if
$f\in \mathcal{I}$ then 
$f=h^jH_j$ where the $H_j$
are 
smooth functions,
which implies $\rho f= f + h^ih^j\partial_{h^j}H_i\implies \rho f-f=h^ih^j\partial_{h^j}H_i
\in\mathcal{I}^2$.
\end{ex} 
$\textbf{Euler vector fields}$ satisfy the following nice properties:
\begin{itemize}
\item Given $I$, the set of $\emph{global}$ Euler vector fields defined on some open neighborhood of $I$ is $\textbf{nonempty}$.
\item For any local Euler vector field $\rho|_U$, for any open set $V\subset U$ there is 
an Euler vector field $\rho^\prime$ defined on a $\textbf{global neighborhood}$ of $I$ such that $\rho^\prime|_V=\rho|_V$. 
\end{itemize} 
\begin{proof}
These two properties result from the fact that one can glue together
Euler vector fields by a partition of unity subordinated
to some cover of some neighborhood $N$ of $I$.
By paracompactness of $M$, we can pick an arbitrary locally finite open cover $\cup_{a\in A} V_a$ of $M$ by open sets $V_a$,
define the subset $J\subset A$ such that for each $a\in J$, $V_a\cap I\neq \emptyset$, 
there is a local chart $(x,h):V_a\mapsto \mathbb{R}^{n+d}$ where the image of $I$ by the local chart is the subspace $\{h^j=0\}$. 
For such charts which have non
empty overlaps with $I$,
we can define an Euler vector field $\rho|_{V_a}$, it suffices to consider the vector field $\rho=h^j\partial_{h^j}$ in each local chart $V_a,a\in J$ and by the example \ref{fundexEuler} this is Euler. 
The vector fields $\rho_a=\rho|_{V_a}$ 
do not necessarily coincide on the overlaps $V_a\cap V_b$. However,
for any partition of unity $(\alpha_a)_a$ subordinated to $(V_a)_a$,  
the vector field $\rho$ defined by the formula
\begin{equation}
\rho=\sum_{a\in J} \alpha_a\rho_a  
\end{equation}
is Euler 
since $ \forall f\in\mathcal{I}(U), \rho f-f= \sum_{a\in J} \alpha_a\rho_a f- \sum_{a\in A} \alpha_a f=\sum_{a\in J}\alpha_a \left(\rho_a f-f \right)-\sum_{a\in A\setminus J}\alpha_a f\in\mathcal{I}^2(U)$
since every $\alpha_a$ for $a\in A\setminus J$ vanishes on some neighborhood
of $I$.
\end{proof}
We can find the general form for all possible Euler vector fields $\rho$ in arbitrary coordinate system $(x,h)$ where $I=\{h=0\}$. 
\begin{lemm}
$\rho|_U$ is $\textbf{Euler}$
if and only if for all $p\in I\cap U$, in $\textbf{any arbitrary}$ local chart $(x,h)$ centered at $p$ where $I=\{h=0\}$, $\rho$ has the standard form 
\begin{equation}\label{generalform}
\rho=h^j\frac{\partial}{\partial h^j} + h^iA_i^j(x,h)\frac{\partial}{\partial x^j} + h^ih^jB_{ij}^k(x,h)\frac{\partial}{\partial h^k}
\end{equation}
where $A,B$ are smooth functions of $(x,h)$.
\end{lemm} 
\begin{proof}
The proof 
is straightforward by noticing
that
\begin{eqnarray}
\forall j , \rho h^j-h^j=o(\Vert h\Vert^2) \\
\forall (i,j), (\rho x^ih^j)-x^ih^j=o(\Vert h\Vert^2),
\end{eqnarray}
from the definition
of $\rho$ being an 
Euler
vector field.
\end{proof}

\subsection{The diffeomorphism invariance
of Euler vector fields.}

The class of Euler vector fields enjoys
many interesting properties,
the first being diffeomorphism
invariance.
From the introduction, let us recall the statement of
Theorem \ref{eulerdiffinvariance}: 
\begin{thm}
Let $M,M^\prime$ be two smooth manifolds,
$I\subset M,I^\prime\subset M^\prime$
smooth embedded submanifolds and $\Phi:=U\longmapsto U^\prime$ 
a local diffeomorphism
such that $\Phi(I\cap U)= I^\prime\cap U^\prime$.
Then
for any
Euler vector field
$\rho$ defined on $U$
the pushforward
$\Phi_*\rho$
is Euler.
\end{thm}

\begin{proof} 
Let $G$ be the pseudogroup of local diffeomorphisms of $M$ (i.e. an element $\Phi$ in $G$ is defined over an open set $U\subset M$ and maps it diffeomorphically to an open set $\Phi(U)\subset M$)
such that
$\forall p\in I\cap U, \forall \Phi\in G, \Phi(p)\in I$.
Then it suffices to establish that
for all Euler vector field 
$\rho$, for all $ \Phi\in G$, $\Phi_*\rho $ is $\textbf{Euler}$.
In the sequel, we shall identify vector fields $X$ with the associated Lie derivative
$L_X$ acting on functions,
then the identity $\forall f\in C^\infty(U), (\Phi_*\rho)f=\Phi^{-1*}(\rho (\Phi^*f))$
holds true, it follows from the well--known functorial identity $\Phi_*\left(\rho f \right)=\left(\Phi_*\rho\right)\left(\Phi_*f\right)$ \cite[Proposition 2.80 p.~93]{LeeDiff}.
Now if we choose $f$ to be an arbitrary function in $\mathcal{I}$
then we get 
\begin{eqnarray}
\forall \Phi\in G,\forall f\in \mathcal{I}, \left(\Phi_*\rho\right)f-f
=\Phi^{-1*} \left(\rho \left(\Phi^*f\right)- \left(\Phi^*f\right)\right).
\end{eqnarray}
Since $\Phi(I)\subset I$, we have actually 
$\Phi^*f\in\mathcal{I}$ hence $\left(\rho \left(\Phi^*f\right)- \left(\Phi^*f\right)\right)\in\mathcal{I}^2$ and we deduce
that
$$\Phi^{-1*} \left(\rho \left(\Phi^*f\right)- \left(\Phi^*f\right)\right)\in\Phi^{-1*}\mathcal{I}^2=\mathcal{I}^2.$$
\end{proof}

\subsection{Local conjugations of scalings.}
We work at the level of germs, a germ of Euler vector field
at $p$ is some Euler vector field defined on some neighborhood of $p$.
A germ
of diffeomorphism (resp smooth family of germs) at $p$ fixing $p$ is some smooth map
$\Phi\in C^\infty(U,M)$ (resp $\Phi\in C^\infty([0,1]\times U,M)$) where $U$ is some neighborhood of $p$, assume
there is a coordinate chart 
$(x^i,h^j)_{1\leqslant i\leqslant n,1\leqslant j\leqslant d}:U\mapsto \mathbb{R}^{n+d}$
such that $I\cap U=\{h^j=0,1\leqslant j\leqslant d\}$ and
$\vert\det d_{x,h}\Phi\vert > 0$ (resp $\inf_{\lambda\in[0,1]}\vert\det d_{x,h}\Phi(.,\lambda)\vert > 0$) 
on $U$.
On the one hand, we saw that the class of Euler vector fields is invariant by the action of $G$,
on the other hand
we will prove
that for any two germs of Euler vector fields $\rho_1,\rho_2$ at $p$, 
there is a germ of diffeomorphism $\Psi$ at $p$ such that $\Psi_*\rho_1=\rho_2$. 

Denote by $S(\lambda)=e^{\log\lambda \rho}$ the scaling operator 
defined by the Euler $\rho$,
$S(\lambda)$ satisfies the identity $S(\lambda_1)\circ S(\lambda_2)=S(\lambda_1\lambda_2)$.
\begin{prop}\label{propositionvariablefamily}
Let $p$ in $I$,  $\rho_1,\rho_2$ be two germs of Euler vector fields at $p$ and $S_a(\lambda)=e^{\log\lambda \rho_a},a=1,2$ 
the corresponding scalings.
Then there is a smooth family $(\Phi(\lambda))_{\lambda\in[0,1]}$ of germs of diffeomorphisms at $p$
such that:
$$ S_2(\lambda)=  S_1(\lambda)\circ \Phi(\lambda) .$$ 
\end{prop}
\begin{proof}
We use a local chart $(x,h):U\mapsto \mathbb{R}^{n+d}$ centered at $p$, where
$I=\{h=0\}$.
We set $\rho=h^j\partial_{h^j}$ which generates the flow 
$S(\lambda)=e^{\log\lambda\rho}$ and we
construct two germs of diffeomorphisms $\Phi_a(\lambda), a=1,2$ at $p$ 
such that
\begin{eqnarray}\label{phia}
\forall \lambda\in (0,1],\Phi_a(\lambda)=S_a^{-1}(\lambda)\circ S(\lambda), a=1,2.
\end{eqnarray}
Then the germ of diffeomorphism $\Phi(\lambda)=\Phi_1(\lambda)\circ \Phi^{-1}_2(\lambda)$
is a solution of our problem.

Let us construct $\Phi_a(\lambda)$ as a solution 
of the differential equation obtained by differentiating 
\ref{phia}:
\begin{eqnarray}
\lambda\frac{\partial}{\partial\lambda}\Phi_a(\lambda)=\left(\rho - S^{-1}(\lambda)_{*}\rho_a \right)\left(\Phi_a(\lambda)\right) \text{ with } \Phi_a(1)=Id 
\end{eqnarray}
Let $f$ be a smooth function and $X$ a vector field, then the pushforward of $fX$ by a diffeomorphism $\Phi$ is:
\begin{equation}\label{pushforwardgeneral}
\Phi_*\left(fX\right)=\left(\Phi_*f\right)\left(\Phi_*X\right). 
\end{equation} 
We use the general form (\ref{generalform}) for Euler vector fields:  
$$\rho_a=h^j\frac{\partial}{\partial h^j} + h^iA_i^j(x,h)\frac{\partial}{\partial x^j} + h^ih^jB_{ij}^k(x,h)\frac{\partial}{\partial h^k}$$ 
hence we apply formula (\ref{pushforwardgeneral}): $$S^{-1}(\lambda)_{*}\rho_a
=\lambda h^j \lambda^{-1}\partial_{h^j}+ \lambda h^iA_i^j(x,\lambda h)\frac{\partial}{\partial x^j} + \lambda^2h^ih^jB_{ij}^k(x,\lambda h)\lambda^{-1}\frac{\partial}{\partial h^k} $$
$$=h^j\partial_{h^j}+ \lambda h^iA_i^j(x,\lambda h)\frac{\partial}{\partial x^j} + \lambda h^ih^jB_{ij}^k(x,\lambda h)\frac{\partial}{\partial h^k}$$
$$\implies \rho-S^{-1}_{*}(\lambda)\rho_a =-\lambda\left( h^iA_i^j(x,\lambda h)\frac{\partial}{\partial x^j} + h^ih^jB_{ij}^k(x,\lambda h)\frac{\partial}{\partial h^k}\right).$$ 
If we define the vector field $X(\lambda)= -\left( h^iA_i^j(x,\lambda h)\frac{\partial}{\partial x^j} + h^ih^jB_{ij}^k(x,\lambda h)\frac{\partial}{\partial h^k}\right)$ then 
\begin{equation}\label{ODEhomotopy}
\frac{\partial \Phi_a}{\partial\lambda}(\lambda)=X\left(\lambda,\Phi_a(\lambda)\right) 
\text{ with }  \Phi_a(1)=Id. 
\end{equation}
$\Phi_a(\lambda)$ satisfies a non autonomous ODE where the vector field 
$X(\lambda,.)$ 
depends smoothly on $(\lambda,x,h)$.
Note that for all $\lambda\in[0,1]$, the 
vector field $X(\lambda)$ vanishes at $p$, therefore
by choosing some sufficiently small open neighborhood $U$ of $p$,
there is a smooth map $\Phi(\lambda,p)$ which integrates the differential
equation
(\ref{ODEhomotopy}) on the interval $[0,1]$.
\end{proof}
We keep the notations and assumptions of the above proposition, 
we give a simple proof of Theorem \ref{eulerlocconjugate} 
which states that Euler vector fields 
are always locally conjugate:
\begin{thm}
Let $\rho_1,\rho_2$ be two germs of Euler vector fields
at $p\in I$.
Then there is a germ of diffeomorphism  
$\Psi$ at $p$ such that
$\rho_1=\Psi_* \rho_2 $.
\end{thm}

\begin{proof}
To prove the above claim, it suffices
to construct $\Psi$ in such a way that 
$S_1(\lambda)=\Psi\circ S_2(\lambda) \circ \Psi^{-1}$ where $S_a(\lambda)=e^{\log\lambda \rho_a},a=(1,2)$.
In local coordinates $(x^i,h^j)_{ij}$ around $p$ where $I=\{h=0\}$, let $\rho=h^j\partial_{h^j}$ be some 
Euler vector field (canonically associated to the choice of coordinates), 
$S(\lambda)=e^{\log\lambda \rho}$ the corresponding scaling and $\Phi_a(\lambda)$ the family of diffeomorphisms $\Phi_a(\lambda)=S^{-1}(\lambda)\circ S_a(\lambda)$ 
which
has a $\textbf{smooth limit}$ $\Psi_a=\Phi_a(0)$ 
when $\lambda\rightarrow 0$ by Proposition \ref{propositionvariablefamily}.
Start from the identity:
$$\Phi_a(\lambda)\circ S(\mu)=\left(S_a^{-1}(\lambda)\circ S(\lambda)\right)\circ S(\mu)=S_a^{-1}(\lambda)\circ S(\lambda\mu)$$ $$=S_a(\mu)\circ S_a^{-1}(\lambda\mu)\circ S(\lambda\mu)=S_a(\mu)\circ\Phi_a(\lambda\mu),$$ 
Hence $\forall (\lambda,\mu), \Phi_a(\lambda)\circ S(\mu)=S_a(\mu)\circ\Phi_a(\lambda\mu)\implies  \Phi_a(0)\circ S(\mu)=S_a(\mu)\circ\Phi_a(0)$ at the limit when $\lambda\rightarrow 0$ where the limit makes sense because $\Phi_a$ is smooth in $\lambda$ at $0$. 
Hence we find that $S_1(\lambda)=\Psi_1\circ \Psi^{-1}_2\circ S_2(\lambda)\circ \Psi_2\circ\Psi_1^{-1}$ and the germ
of diffeomorphism $\Psi=\Psi_1\circ \Psi^{-1}_2$ solves our problem.
\end{proof}

\section{The space $E_{s,I}(U)$.}
\label{Esiuconstructionsection}
In this section, $I$ is a closely
embedded submanifold of $M$ and we use
Euler vector fields 
to scale distributions along $I$ and to define
scale spaces of distributions.
First a set
$U$ is called $\rho$-convex 
if $U$ is stable by the flow 
$\left(e^{-t\rho}\right)_{t>0}$.
We give a definition
of weakly homogeneous distributions on manifolds
but this definition is $\rho$
dependent:
\begin{defi}
Let $U$ be a $\rho$-convex open set. The set 
$E^\rho_s(U)$ is defined as the set of distributions $t\in \mathcal{D}^\prime(U)$ such that
$$\forall\varphi\in \mathcal{D}(U),\exists C, \sup_{\lambda\in(0,1]} \vert\left\langle\lambda^{-s}t_\lambda,\varphi\right\rangle \vert\leqslant C.$$ 
\end{defi}
We 
next define the space 
$E_{s,p}^\rho$ of distributions which are
locally weakly homogeneous of degree $s$ 
at $p\in I$.
\begin{defi}
A distribution $t$ belongs to $E_{s,p}^\rho$ if
there exists
an open $\rho$-convex set $U\subset M$ such that $\overline{U}$ is a neighborhood of $p$ and such that $t\in E_s^{\rho}(U)$.
\end{defi}
\paragraph{A key locality theorem.}
The next Theorem proves a crucial result
that 
if $t$ is locally 
$E_{s,p}^\rho$ for some Euler 
vector field $\rho$ then
it is locally
$E_{s,p}^\rho$ for \textbf{all Euler vector fields } $\rho$.
\begin{thm}\label{locthmGOOD}
Let $p\in I$, if $t$ belongs to $E_{s,p}^\rho$ for some Euler vector field $\rho$, then it is so
for any other Euler vector field.
\end{thm}
\begin{proof}
It suffices to prove the equality
$E_{s,p}^{\rho_1}=E_{s,p}^{\rho_2}$
for any pair $\rho_1,\rho_2$ of Euler 
vector fields
at $p$.
Recall there is a smooth family of germs $\Phi(\lambda)_\lambda$ satisfying
$\Phi(\lambda)=S_1^{-1}(\lambda)\circ S_2(\lambda)$ where $(S_a(\lambda)=e^{\log\lambda\rho_a})_{a\in\{1,2\}}$, 
by 
Proposition
\ref{propositionvariablefamily}.
Then $\lambda^{-s}S_2(\lambda)^*t=\lambda^{-s}\Phi(\lambda)^*\left(S_1(\lambda)^*t\right)$ is a bounded
family of distribution in $\mathcal{D}^\prime(V)$
for some neighborhood $V$ of $p$ implies that
$\lambda^{-s}\left(S_1(\lambda)^*t\right)$
is also a bounded
family of distribution in $\mathcal{D}^\prime(V^\prime)$
for some smaller neighborhood $V^\prime$ of $p$.
\end{proof}

A comment on the statement of the theorem, first
the definition of $\rho$-convexity is 
$\forall p\in U,\forall \lambda\in(0,1], S(\lambda,p)\in U$, 
the fact that we let $\lambda$ to be positive 
allows $U$ to have \emph{empty intersection} with $I$.
The previous theorem allows to give a definition 
of the spaces of distributions $E_{s,p}$ and $E_{s,I}(U)$ 
which makes
\emph{no mention of the choice of Euler vector field}:
\begin{defi}\label{defEs} 
A distribution $t$ belongs to $E_{s,p}$
if $t$ belongs to $E_{s,p}^\rho$ for some $\rho$. We define $E_{s,I}(U)$ 
as the space of all distributions $t\in\mathcal{D}^\prime(U)$ such that $t\in E_{s,p}^\rho,$
$\forall p\in I\cap int(\overline{U})$.
\end{defi}
An equivalent definition of the space $E_{s,I}(U)$
is the following:
\emph{ $t$ belongs to the
local space $E_{s,I}(U)$ if and only if
for all $p\in  I\cap int(\overline{U})$, there is
some open chart $\psi:V_p\mapsto \mathbb{R}^{n+d}$, 
$\psi(I)\subset \mathbb{R}^n\times\{0\}$
where $\lambda^{-s}(\psi_*t)(x,\lambda h) $
is bounded in $\mathcal{D}^\prime(\psi(V_p\cap U))$.}

It is immediate that
$E_{s,I}$ satisfies the restriction
property: if $V\subset U$ then $p\in int(\overline{V})\cap I\implies p\in int(\overline{U})\cap I$
and therefore 
the restriction of $E_{s,I}(U)$ on $V$
is $E_{s,I}(V)$.

A consequence
of Theorem \ref{locthmGOOD} is the following 
properties of $E_{s,I}$ under 
gluings:
\begin{thm}
$E_{s,I}$ satisfies the following gluing property:
if $\cup_i V_i$ is an open cover of $U$ s.t. 
$\cup_i int\left(\overline{V_i}\right)$ is a neighborhood of $U$ then
for $t\in\mathcal{D}^\prime(\cup V_i)$,
$t\in E_{s,I}(V_i),\forall i\implies t\in E_{s,I}(U)$
\end{thm}
\begin{proof}
It suffices to prove that
$t\in E_{s,p}(U)$ for all 
$p\in int(\overline{U})\cap I$.
Let $p\in int\left(\overline{U}\right)\cap I$, then obviously
$p\in \cup_i int\left(\overline{V_i}\right)$ since $\cup_i int\left(\overline{V_i}\right)$
is a neighborhood of $U$. Then by definition of $t\in E_{s,p}(V_i)$, there is some neighborhood 
$V_p$ of $p$ s.t. $V_p\subset int\left(\overline{V_i}\right)$ for some $i$
and $\lambda^{-s}e^{\log\lambda\rho*}t$ is 
bounded in $\mathcal{D}^\prime(V_p\cap V_i)$ which implies 
in particular that
$\lambda^{-s}e^{\log\lambda\rho*}t$ is 
bounded in $\mathcal{D}^\prime(\tilde{V_p} \cap U)$ where $\tilde{V_p}=V_p\cap int\left(\overline{U}\right)$
is a neighborhood of $p$ in $int\left(\overline{U}\right)$, therefore $t\in E_{s,p}(U)$.
\end{proof}
We prove Theorem \ref{Esdiffinvariance}
which claims that
$E_{s,I}(U)$ satisfies
a property of
diffeomorphism
invariance:
\begin{thm}
Let $I$ (resp $I^\prime$) be a closed embedded submanifold of $M$
(resp $M^\prime$), $U\subset M$ (resp $U^\prime\subset M^\prime$) open and $\Phi:U^\prime\mapsto U$ 
a diffeomorphism s.t. $\Phi(U^\prime\cap I^\prime)=I\cap U$.
Then $\Phi^*E_{s,I}(U)=E_{s,I^\prime}(U^\prime)$.
\end{thm}
\begin{proof}
By Theorem \ref{locthmGOOD}, we can localize the proof
at all points $p\in int\left(\overline{U}\right)\cap I$. Let $p\in int\left(\overline{U}\right)\cap I$,
then $t\in E_{s,I}(U)$ implies
by definition 
that $t\in E_{s,p}^\rho$ for some $\rho$ 
which means that:
$$\lambda^{-s}e^{\log\lambda\rho*}t \text{ bounded in }\mathcal{D}^\prime(V), int(\overline{V})\text{ neighborhood of }p$$
$$\Leftrightarrow \lambda^{-s}\Phi^* e^{\log\lambda\rho*}\Phi^{-1*} (\Phi^*t)
\text{ bounded in }\mathcal{D}^\prime(\Phi^{-1}(V))$$
because the pull--back by a diffeomorphism
is bounded from $\mathcal{D}^\prime(V)$ to $\mathcal{D}^\prime(\Phi^{-1}(V))$~\cite[Prop 6.1]{Viet-wf2},
$$\Leftrightarrow   \lambda^{-s} e^{\log\lambda(\Phi^{-1}_*\rho)*} (\Phi^*t)
\text{ bounded in }\mathcal{D}^\prime(\Phi^{-1}(V)) $$
where the vector field $\Phi^{-1}_*\rho$
is a germ of Euler field near $p$ by 
\ref{eulerdiffinvariance}.
Therefore $\Phi^*t$
is in $E_{s,p^\prime}^{\Phi^{-1}_*\rho}$ where $p^\prime=\Phi^{-1}(p)$
and repeating the proof for all $p\in int\left(\overline{U}\right)\cap I$
yields the claim.
\end{proof}

\section{The extension problem 
on flat space.}
\label{flatspacesection}
\paragraph{Formulation of the problem.}
We work in $\mathbb{R}^{n+d}$ with coordinates $(x,h)$, $I=\mathbb{R}^n\times\{0\}$ is the linear subspace $\{h=0\}$.
In the sequel, unless it is specified otherwise, we will
always assume that we work with open sets $U$ of the form
$U_1\times U_2$ where $U_1$ (resp $U_2$) is an open subset of $\mathbb{R}^n$
(resp $\mathbb{R}^d$) s.t. $\lambda U_2\subset U_2,\forall\lambda\in(0,1]$ in particular
such $U$ is \emph{convex} meaning that: 
\begin{eqnarray}\label{convex}
(x,h)\in U\implies \forall\lambda\in(0,1], (x,\lambda h)\in U.
\end{eqnarray}
We reformulate the extension problem
on flat space:
\begin{defi}
We are given a convex open set $U\subset \mathbb{R}^{n+d}$ and $
I= \mathbb{R}^n\times\{0\}$. We have a distribution 
$t\in\mathcal{D}^\prime(U\setminus I)$
and we would like to
find under what reasonable conditions
on $t$ one can construct an extension
$\overline{t}\in\mathcal{D}^\prime(U)$.
\end{defi}
\subsection{Construction of a formal extension.}
In this subsection,
we construct a candidate
for the formal
extension.
\paragraph{Defining a smooth partition of unity.}
A partition of unity
will provide us with some family
of smooth functions
supported in $\mathbb{R}^{n+d}\setminus I$
approximating the constant function $1\in C^\infty(\mathbb{R}^{n+d}\setminus I)$.
\begin{defi}
A smooth partition 
of unity is a function
$\Psi\in C^\infty\left((0,\infty),\mathbb{R}^{n+d}\setminus I\right)$
such that $\forall\Lambda\in(0,\infty),\Psi_\Lambda=0$ in some neighborhood of $I$
and
$\Psi_\Lambda\underset{\Lambda\rightarrow \infty}{\rightarrow} 1$
for the Fr\'echet topology of $C^\infty(\mathbb{R}^{n+d}\setminus I)$.
\end{defi}
Motivated by the above definition,
we choose a function $\chi$
such that $\chi=1$ in a neighborhood
of $I$ and the projection $\pi:\mathbb{R}^n\times\mathbb{R}^d\longmapsto \mathbb{R}^n\times \{0\}$ is proper on the support
of $\chi$. This implies $\chi$ satisfies the following constraint: 
for all compact set $K\subset \mathbb{R}^{n},\exists (a,b)\in \mathbb{R}^2$ such that $b>a>0$ and $\chi|_{(K\times\mathbb{R}^d)\cap \{\vert h\vert\leqslant a\}}=1, \chi|_{(K\times\mathbb{R}^d)\cap \{\vert h\vert\geqslant b\}}=0$.
We set $\Psi(\Lambda,x,h)=1-\chi(x,\Lambda h)$ and
it is a simple exercice to verify that this defines
a partition of unity
of $\mathbb{R}^{n+d}\setminus I$.
\paragraph{A candidate for the extension.}
\begin{prop}
Let $U$ be an open set of $\mathbb{R}^{n+d}$, if
$t\in\mathcal{D}^\prime(U\setminus I)$ then
for any smooth partition of unity $\Psi_\Lambda$
\begin{eqnarray}
t=\lim_{\Lambda\rightarrow +\infty} t\Psi_\Lambda
\end{eqnarray}
as distribution on $U\setminus I$.
\end{prop}

From the above proposition, we deduce that
if $\lim_{\Lambda\rightarrow +\infty} t\Psi_\Lambda$
\textbf{converges in} $\mathcal{D}^\prime(U)$
the limit defines an extension
of $t$.
So this raises the question,
for all test function $\varphi\in \mathcal{D}(U)$, does
the limit 
$\lim_{\Lambda\rightarrow +\infty}\left\langle t,\Psi_\Lambda\varphi \right\rangle$
exist ?
To study this question,
we introduce a continuous decomposition
of our partition of unity
formula
\begin{eqnarray*}
1-\chi(x,\Lambda h)&=&1-\chi(x,h)+\chi(x,h)-\chi(x,\Lambda h)\\
&=& \int_{\Lambda^{-1}}^1 \frac{d\lambda}{\lambda}\lambda\frac{d}{d\lambda}\left[\chi(x,\lambda^{-1}h)\right]+1-\chi(x,h)\\
&=& \int_{\Lambda^{-1}}^1\frac{d\lambda}{\lambda} \left(-\rho\chi\right)(x,\lambda^{-1}h)+1-\chi(x,h),
\end{eqnarray*}
where $\rho$ is the Euler vector field
$\sum_{j=1}^dh^j\frac{\partial}{\partial h^j}$
which scales tranversally to $I=\{h^j=0,1\leqslant j\leqslant d\}$.
In the sequel, we
will 
write $\rho=h\frac{\partial}{\partial h}$
for brevity.
Set $\psi=-\rho\chi$ and define
the scaling by a factor $\lambda\in(0,1]$:
\begin{eqnarray*}
\psi_\lambda(x,h)=\psi(x,\lambda h).
\end{eqnarray*}
In these notations,
the partition of
unity formula simply writes:
\begin{eqnarray}\label{partitionofunityformula}
\chi-\chi_{\Lambda}=\int_{\Lambda^{-1}}^1\frac{d\lambda}{\lambda} \psi_{\lambda^{-1}}.
\end{eqnarray}
In the sequel,
instead
of studying the
limit
$\Lambda\rightarrow\infty$,
we will set
$\varepsilon^{-1}=\Lambda$
and study instead
the limit $\varepsilon\rightarrow 0$.
We will also denote by $\pi$
the projection $(x,h)\in\mathbb{R}^n\times\mathbb{R}^d\longmapsto x\in\mathbb{R}^n$
and $\chi$ will always designate
a smooth function
such that $\chi=1$ in some neighborhood of $I$ and 
$\pi$
is proper on $\text{supp }\chi$.
In what follows, we will
study
the behaviour of
\begin{eqnarray}
t(1-\chi_{\varepsilon^{-1}})=\int_{\varepsilon}^1\frac{d\lambda}{\lambda} t\psi_{\lambda^{-1}}+t(1-\chi).
\end{eqnarray}
when $\varepsilon\rightarrow 0$.

\subsection{The extension theorems.}

\subsubsection{Some definitions and notations.}
Let us introduce
the terminology
needed to state our theorems.
Let $U\subset \mathbb{R}^{n+d}$ be an open set,
we denote by $(x,h;\xi,\eta)$ 
the coordinates 
in cotangent space $T^* U$, 
where $\xi$ (resp $\eta$) is dual to $x$ (resp $h$).
$T^\bullet U$ denotes the cotangent $TU$ 
minus the zero section $\underline{0}$. 

A set $\Gamma\subset T^\bullet U$ 
is \emph{stable by scaling}
if 
\begin{eqnarray}\label{stabilitycondition}
\forall\lambda\in(0,1], \left(\{(x,\lambda^{-1}h;\xi,\lambda\eta) 
; (x,h;\xi,\eta)\in\Gamma   \}
\cap T^\bullet U\right) \subset \Gamma.
\end{eqnarray}
Concisely, if we denote by
$\Phi_\lambda^*\Gamma$
the pull--back of $\Gamma$
by $\Phi_\lambda$ \cite{Viet-wf2}, we require
that $\forall\lambda\in(0,1],
\left(\Phi_\lambda^*\Gamma\cap T^\bullet U\right)\subset \Gamma$.
We also denote by
$T^*_I\mathbb{R}^{n+d}$
the restriction
of $T^*\mathbb{R}^{n+d}$
on $I$ and $N^*(I)$ 
the conormal bundle of $I$.
As we explained in the introduction, the extension theorem
has no positive solution
for arbitrary distributions
in $U\setminus I$.
However, if we impose
that the distribution
has ``moderate growth''
in terms of scaling
then we will be able 
to solve it.
The scaling of
distribution
is defined
by duality
\begin{eqnarray*}
\forall \varphi\in\mathcal{D}(U),
\left\langle t_\lambda, \varphi \right\rangle
=\lambda^{-d} \left\langle t, \varphi_{\lambda^{-1}} \right\rangle\\
\text{where }\varphi_{\lambda^{-1}}=\varphi(x,\lambda^{-1} h).
\end{eqnarray*}

In the sequel, for a given open set $U$, 
a compact set $K\subset U$, 
we will denote by
$\left(\pi_{m,K}\right)_{m\in\mathbb{N}}$ 
the collection of continuous
seminorms on the Fr\'echet space
$\mathcal{D}_K(U)$ of test functions supported
on $K$ defined
as:
\begin{eqnarray*}
\forall \varphi\in\mathcal{D}_K(U), \pi_{m,K}(\varphi)=\sup_{\vert\alpha\vert\leqslant m,x\in K} \vert\partial^\alpha\varphi(x)\vert.
\end{eqnarray*}

\paragraph{Weakly homogeneous distributions
in $\mathcal{D}^\prime_\Gamma$.}
Let us formalize
this notion of distribution
having nice behaviour 
under scaling by
defining 
the main space of
distributions 
for which the
extension problem
has a positive answer.

Using the recent work \cite[6.3]{dabrowski2013functional},
we can characterize
bounded sets in $\mathcal{D}^\prime_\Gamma$
by duality pairing.
A set $B\subset \mathcal{D}^\prime_\Gamma(U)$
is bounded if for every $v\in \mathcal{E}^\prime_\Lambda(U)$ where
$\Lambda$ is an open cone s.t. $\Lambda\cap-\Gamma=\emptyset$,
we have
\begin{eqnarray*}
\sup_{t\in B}\vert\left\langle t,v \right\rangle  \vert<+\infty.
\end{eqnarray*}

\begin{defi}
Let 
$\Gamma\subset T^\bullet U$ be 
a closed conic set stable by scaling.
A distribution $t$
is weakly homogeneous of degree $s$
in $\mathcal{D}^\prime_\Gamma(U)$, if
for all distribution 
$v\in \mathcal{E}^\prime_\Lambda(U)$
where $\Lambda=-\Gamma^c$,
\begin{eqnarray*}
\sup_{\lambda\in[0,1]}
\vert\left\langle\lambda^{-s}t_\lambda,v\right\rangle\vert<+\infty.
\end{eqnarray*}
We denote this space
by $E_{s}(\mathcal{D}^\prime_\Gamma(U))$
and we endow
it with the 
locally convex topology
generated
by the seminorms
\begin{equation*}
P_{B}(t)=\underset{\lambda\in[0,1],v\in B}{\sup}
\vert\left\langle\lambda^{-s}t_\lambda,v\right\rangle\vert
\end{equation*}
for $B$ equicontinuous \cite[lemma 6.3]{Viet-wf2}
in $\mathcal{E}_\Lambda^\prime$.
\end{defi}

We recover the definition
of Yves Meyer in the particular case where
$\Gamma=T^\bullet \mathbb{R}^{n+d}$ in which case
$\mathcal{D}^\prime_\Gamma=\mathcal{D}^\prime$.

A key conceptual step
in our approach
is to think of $\lambda^{-s}t(x,\lambda h)$
as a distribution
of the three variables
$(\lambda,x,h)$.
Let us define the map
\begin{eqnarray}
\Phi:(\lambda,x,h)\in\mathbb{R}\times\mathbb{R}^{n+d}\longmapsto (x,\lambda h)\in\mathbb{R}^{n+d}.
\end{eqnarray}

\begin{thm}\label{thm1}
Let $s\in\mathbb{R}$ s.t. $s+d> 0$, 
$\Gamma\subset T^\bullet U$ 
a closed conic set stable 
by scaling.
If $t\in\mathcal{D}^\prime(U\setminus I)$ is weakly
homogeneous of degree $s$ in 
$\mathcal{D}^\prime_\Gamma(U\setminus I)$,
then $\overline{t}=\lim_{\varepsilon\rightarrow 0}
t(1-\chi_{\varepsilon^{-1}})$
is a well defined extension
of $t$
and $WF(\overline{t})\subset 
WF(t)\cup N^*(I)\cup \Xi$
where 
\begin{equation*}
\Xi=\{(x,0;\xi,\eta) |\exists (x,h;\xi,0)\in\Gamma\cap T_{\text{supp }\psi}^*U  \}.
\end{equation*}

\end{thm}
Before we prove the theorem,
let us show why the set
$WF(t)\cup N^*(I)\cup \Xi$
is a closed conic set.
Recall that $U=U_1\times U_2\subset \mathbb{R}^n\times\mathbb{R}^d$
where $\lambda U_2\subset U_2,\forall\lambda\in(0,1]$. We may assume w.l.o.g
that $U_2$ contains a set of the form $\{0<\vert h\vert\leqslant \varepsilon\}$.
There is nothing to prove
over $U\setminus I$
since $WF(t)$ is closed
in $T^\bullet(U\setminus I)$,
therefore we study the closure
of $WF(t)\cup N^*(I)\cup \Xi$
in $T_I^*U$.
Let $(x,0;\xi,\eta)$
be in its closure
$\overline{WF(t)\cup N^*(I)\cup \Xi}$.
If $\xi=0$ then
$(x,0;0,\eta)\in N^*(I)$.
Otherwise $\xi\neq 0$, 
there is a sequence
$(x_n,h_n;\xi_n,\eta_n)\rightarrow 
(x,0;\xi,\eta)$
where $(x_n,h_n;\xi_n,\eta_n)\in WF(t)$
and $h_n\rightarrow 0$.
But since $WF(t)\subset \Gamma$
and since $\Gamma$
is scale invariant
then
$(x_n,\varepsilon\frac{h_n}{\vert h_n\vert};\xi_n,\varepsilon^{-1}\vert h_n\vert\eta_n)\in 
\Gamma$. By compactness
of the unit sphere,
we can extract
a convergent subsequence for
$\varepsilon\frac{h_n}{\vert h_n\vert}$
and the limit $(x,h;\xi,0)$
is in $\Gamma$. 
Therefore by definition
of $\Xi$, we will have
$(x,0;\xi,\eta)\in\Xi$
and this implies that
$WF(t)\cup N^*(I)\cup \Xi$
is closed.

\begin{proof}

We have to establish the convergence
of $t(1-\chi_{\varepsilon^{-1}})$ in 
$\mathcal{D}_\Lambda^\prime(U)$
when $\varepsilon\rightarrow 0$ for $\Lambda=WF(t)\cup N^*(I)\cup \Xi$.
Our proof is divided in three parts, in the first,
we prove that the limit exists
in $\mathcal{D}^\prime(U)$ with arguments
similar to \cite{Meyer-98} but in our setting
of continuous partition of unity.
Then in the second part, 
we derive a new integral formula
for $t(1-\chi_{\varepsilon^{-1}})$, and
we shall
use the integral formula
to show that the family
$\left(t(1-\chi_{\varepsilon^{-1}})\right)_\varepsilon$
is bounded in $\mathcal{D}^\prime_\Lambda(U)$
using the behaviour of the WF under
the fundamental operations on distributions 
\cite{Viet-wf2}. Then $\underset{\varepsilon\rightarrow 0}{\lim} t(\chi-\chi_\varepsilon)$
converges in $\mathcal{D}^\prime(U)$ and is bounded in $\mathcal{D}^\prime_\Lambda(U)$
implies that $\underset{\varepsilon\rightarrow 0}{\lim} t(\chi-\chi_\varepsilon)$
converges in $\mathcal{D}^\prime_\Lambda(U)$.

\textbf{Step 1.}
We prove that
$\underset{\varepsilon\rightarrow 0}{\lim}
t(1-\chi_{\varepsilon^{-1}})$
exists in $\mathcal{D}^\prime(U)$ when 
$\varepsilon\rightarrow 0$.
Let us give a
different analytical expression
using
the partition
of unity formula,
\begin{eqnarray*}
t(1-\chi_{\varepsilon^{-1}})
=\int_{\varepsilon}^1 \frac{d\lambda}{\lambda}
t\psi_{\lambda^{-1}}+
t(1-\chi).
\end{eqnarray*}
Therefore:
\begin{eqnarray}\label{Meyerformula}
\left\langle t(\chi-\chi_{\varepsilon^{-1}}),\varphi\right\rangle
&=&\int_{\varepsilon}^1 d\lambda
\lambda^{s+d-1} \left\langle
(\lambda^{-s}t_\lambda)\psi,
\varphi_\lambda \right\rangle.
\end{eqnarray}
It follows that the r.h.s of \ref{Meyerformula} has a limit
when $\varepsilon\rightarrow 0$ since $\lambda^{s+d-1}$
is integrable on $[0,1]$.
It remains to prove that the limit is a distribution.
$(\lambda^{-s}t_\lambda)_{\lambda\in(0,1]}$ is bounded in $\mathcal{D}^\prime(U\setminus I)$ therefore
for all compact subset $K\subset U\setminus I$:
\begin{eqnarray*}
\exists C_K, \forall\varphi\in\mathcal{D}_K(U), \underset{\lambda\in(0,1]}{\sup} \vert\left\langle\lambda^{-s}t_\lambda,\varphi \right\rangle\vert\leqslant C_K\pi_{m,K}(\varphi).
\end{eqnarray*}
For all compact subset
$K^\prime\subset \mathbb{R}^{n+d}$ 
and for all $\varphi\in\mathcal{D}_{K^\prime}(U)$, 
the family
$(\psi\varphi_\lambda)_\lambda$ has 
fixed compact support which does not meet
$I$ and is bounded in
$\mathcal{D}_{K}(U\setminus I)$ for some compact
set $K$:
\begin{eqnarray*}
\forall\lambda\in(0,1], \pi_{m,K}( \psi\varphi_\lambda )\leqslant C_2\pi_{m,K^\prime}(\varphi).
\end{eqnarray*}
The two above bounds
easily imply
that:
\begin{eqnarray*}
\forall\varphi\in\mathcal{D}_{K^\prime}(U),\,\ \underset{\lambda\in(0,1]}{\sup} \vert\left\langle\lambda^{-s}t_\lambda, \psi\varphi_\lambda\right\rangle\vert &\leqslant & C_KC_2\pi_{m,K^\prime}(\varphi)\\
\implies \vert \left\langle t(\chi-\chi_{\varepsilon^{-1}}),\varphi\right\rangle\vert
&\leqslant &\vert\int_{\varepsilon}^1 d\lambda
\lambda^{s+d-1} \left\langle
(\lambda^{-s}t_\lambda)\psi,
\varphi_\lambda \right\rangle\vert\\
&\leqslant &
\frac{1-\varepsilon^{s+d}}{s+d}C_KC_2\pi_{m,K^\prime}(\varphi)\\
\implies \lim_{\varepsilon\rightarrow 0} \vert \left\langle t(\chi-\chi_{\varepsilon^{-1}}),\varphi\right\rangle\vert
&\leqslant & \frac{C_KC_2\pi_{m,K^\prime}(\varphi)}{s+d}
\end{eqnarray*}
The above bound means that
$\underset{\varepsilon\rightarrow 0}{\lim} t(\chi-\chi_{\varepsilon^{-1}})$
is well defined
in $\mathcal{D}^\prime(U)$. 
But the difficult 
point is to control the wave front set of the
limit over the 
subspace $I=\{h=0\}$.

\textbf{Step 2} We just proved
that $\underset{\varepsilon\rightarrow 0}{\lim} t(\chi-\chi_\varepsilon)$
converges in $\mathcal{D}^\prime(U)$. In order to control the WF
of the limit, it suffices to prove that the family
$ t(\chi-\chi_\varepsilon)_\varepsilon$ is bounded
in $\mathcal{D}^\prime_\Lambda(U),\Lambda=WF(t)\cup N^*(I)\cup \Xi$. 
We propose
a simple method which consists
in giving a new integral
formula for the identity \ref{Meyerformula}.
We double
the space $\mathbb{R}^{n+d}$ and transform
the  formula $\int_{\varepsilon}^1 \frac{d\lambda}{\lambda}
\lambda^{s+d} \left\langle
(\lambda^{-s}t_\lambda)\psi,
\varphi_\lambda \right\rangle$ into an integral formula
on $\mathbb{R}\times\mathbb{R}^{n+d}\times\mathbb{R}^{n+d}$.
We work in $\mathbb{R}\times\mathbb{R}^{n+d}\times\mathbb{R}^{n+d}$
with coordinates $(\lambda,x,h,x^\prime,h^\prime)$. We denote
by $\delta\in\mathcal{D}^\prime(\mathbb{R}^{n+d})$ the delta distribution
supported at $(0,0)\in\mathbb{R}^{n+d}$ and $\delta_\Delta(.,.)$ the 
delta distribution supported
by the diagonal $\Delta\subset\mathbb{R}^{n+d}\times \mathbb{R}^{n+d}$
where we have the relation
$\delta_\Delta((x,h),(x^\prime,h^\prime))=\delta(x-x^\prime,h-h^\prime)$.  
Thus $\left\langle t(\chi-\chi_{\varepsilon^{-1}}),\varphi\right\rangle$
\begin{eqnarray*}
&=&\int_{\varepsilon}^1 \frac{d\lambda}{\lambda}
\lambda^{s+d} \left\langle
(\lambda^{-s}t_\lambda)\psi,
\varphi_\lambda \right\rangle\\
&=&\int_{\mathbb{R}^{n+d}}dx^\prime dh^\prime \int_{\mathbb{R}\times\mathbb{R}^{n+d}}
\frac{d\lambda}{\lambda} dxdh 1_{[\varepsilon,1]}(\lambda)\lambda^{s+d}\lambda^{-s}t(x,\lambda h)\psi(x,h)\delta(x-x^\prime,\lambda h-h^\prime)\varphi(x^\prime,h^\prime). 
\end{eqnarray*}
Finally, we end up
with the integral formula:
\begin{eqnarray}\label{fundintegralformula}
t(\chi-\chi_{\varepsilon^{-1}})(x^\prime,h^\prime)&=&
\int_{\mathbb{R}\times\mathbb{R}^{n+d}}
d\lambda dxdh 1_{[\varepsilon,1]}(\lambda)\lambda^{s+d-1}\lambda^{-s}t(x,\lambda h)\psi(x,h)\delta(x-x^\prime,\lambda h-h^\prime).
\end{eqnarray}

It suffices to
estimate $\Lambda$ 
over $I$
since we already know that
the family $t(\chi-\chi_{\varepsilon^{-1}})_\varepsilon$
is bounded in $\mathcal{D}_{WF(t)}^\prime(U\setminus I)$ i.e. $\Lambda\cap T^*\left(U\setminus I\right)=WF(t)$. 
We want to calculate the WF of the r.h.s of
(\ref{fundintegralformula}) in $T^*_IU$. 
\begin{enumerate}
\item decompose the r.h.s of (\ref{fundintegralformula})
in two blocks
\begin{eqnarray*}
\underset{B_{1,\varepsilon}}{\underbrace{1_{[\varepsilon,1]}(\lambda)\lambda^{s+d-1}\lambda^{-s}t(x,\lambda h)\psi(x,h)}}
\underbrace{\delta(x-x^\prime,\lambda h-h^\prime)}
\end{eqnarray*} 
\item $1_{[\varepsilon,1]}(\lambda)\lambda^{s+d-1}\in L^1(\mathbb{R})$ and 
$t\in E_s(\mathcal{D}^\prime_\Gamma(U))$ hence by Lemma \ref{propfamily2} proved in appendix, the block
$\left(B_{1,\varepsilon}=1_{[\varepsilon,1]}(\lambda)\lambda^{s+d-1}\lambda^{-s}t(x,\lambda h)\psi(x,h)\right)_{\varepsilon}$
is a bounded family in
$\mathcal{D}_{V}^\prime(\mathbb{R}\times U)$ when
$\varepsilon\in(0,1]$ and
where 
\begin{eqnarray}
V=\{
\left(\begin{array}{ccc}
\lambda & ;&\widehat{\lambda} \\
x &;& \widehat{\xi}\\
h &;&\widehat{\eta}  
\end{array}\right)
| \left(\begin{array}{ccc}
x &;& \widehat{\xi}\\
h &;& \widehat{\eta}  
\end{array}\right)\in  \Gamma\cup \underline{0}, (x,h)\in\text{supp }\psi\}.
\end{eqnarray}
\end{enumerate}

We evaluate the wave front set of the family
of products
of distributions $\left(B_{1,\varepsilon}(\lambda,x,h)\delta(x-x^\prime,\lambda h-h^\prime)\right)_\varepsilon$
in $T^*(\mathbb{R}\times U\times \mathbb{R}^{n+d})$ using the functional properties of the H\"ormander
product \cite[Theorem 7.1]{Viet-wf2}. 
We start with
the wave front set of the
various distributions involved
in formula (\ref{fundintegralformula}), the family 
$B_{1,\varepsilon}(\lambda,x,h)\otimes 1(x^\prime,h^\prime)$ is bounded in $\mathcal{D}^\prime_{\Lambda_1}(\mathbb{R}\times\mathbb{R}^{n+d}\times\mathbb{R}^{n+d})$
where:
\begin{eqnarray*}
&&\Lambda_1=\{ \left(\begin{array}{ccc}
\lambda & ;&\widehat{\lambda} \\
x & ;& \widehat{\xi}\\
h& ;&\widehat{\eta}\\
x^\prime &;& 0 \\
h^\prime &;& 0  
\end{array}\right)
| \left(\begin{array}{ccc}
x &;& \widehat{\xi}\\
h &;& \widehat{\eta}  
\end{array}\right)\in  \Gamma\cup \underline{0}, (x,h)\in\text{supp }\psi
\}\\
WF\left(\delta_\Delta(\Phi,.)\right)&\subset &\Lambda_2=\{\left(
\begin{array}{ccc}
\lambda &;& -\left\langle h,\eta \right\rangle\\
x &;& -\xi\\
h &;& -\lambda\eta\\
x^\prime &;&\xi\\
h^\prime &;&\eta
\end{array}
\right) 
| (x,\lambda h)=(x^\prime,h^\prime) \text{ and }(\xi,\eta)\neq (0,0)  \}.
\end{eqnarray*}
Note that $\Lambda_1
\cap -\Lambda_2=\emptyset$ which
implies by hypocontinuity of the H\"ormander
product \cite[Theorem 7.1]{Viet-wf2} that
the products
$\left(B_{1,\varepsilon}(\lambda,x,h)\delta(x-x^\prime,\lambda h-h^\prime)\right)_\varepsilon$
are bounded in $\mathcal{D}^\prime_{\Lambda_1+\Lambda_2\cup \Lambda_1\cup\Lambda_2}$.

The projection
\begin{eqnarray*}
\pi_3:=(\lambda,x,h,x^\prime,h^\prime)\longmapsto (x^\prime,h^\prime)
\end{eqnarray*}
is proper on the support of $u$ therefore
the pushforward of $ B_{1,\varepsilon}(\lambda,x,h)\delta(x-x^\prime,\lambda h-h^\prime)$
by $\pi_3$, which equals the integral
$\int_{\mathbb{R}\times U} d\lambda dxdh  B_{1,\varepsilon}(\lambda,x,h)\delta(x-x^\prime,\lambda h-h^\prime)$, 
exists in the
distributional sense.
By continuity hence boudedness of the pushforward \cite[Theorem 7.3]{Viet-wf2},
we find that the family
$\left(t(\chi-\chi_{\varepsilon^{-1}})\right)_\varepsilon$ is bounded
in $\mathcal{D}^\prime_\Lambda$ where
\begin{eqnarray}\label{formula2}
\left(\Lambda \cap T^*_IU\right)
\subset \pi_{3*}\left(\left(\Lambda_1+\Lambda_2\right)\cup\Lambda_1\cup\Lambda_2 \right)\cap T^*_IU .
\end{eqnarray}
We study the closed conic set  
$\pi_{3*}\left(\Lambda_1+\Lambda_2\right)\cap T^*_IU$:
\begin{eqnarray*}
&(x^\prime,0;\xi,\eta)\in \pi_{3*}\left(\Lambda_1+\Lambda_2 \right)\\
\Leftrightarrow & 
\left\lbrace 
\begin{array}{c}
\widehat{\lambda}-\left\langle h,\eta \right\rangle=0\\
\widehat{\xi}-\xi=0\\
\widehat{\eta}-\lambda\eta=0
\end{array}
\text{ s.t. }
(x,\lambda h)=(x^\prime,0), (x,h)\in\text{supp }\psi,
\left( \begin{array}{ccc} x&;& \widehat{\xi} \\
h &;& \widehat{\eta} \end{array}\right)\in\Gamma\cup\underline{0}\right\rbrace
\end{eqnarray*}
has a solution. Note that
$\left\lbrace\begin{array}{c}(x,h)\in\text{supp }\psi \\ 
(x,\lambda h)=(x^\prime,0)\\
\widehat{\eta}-\lambda\eta=0
\end{array}\right\rbrace\implies \vert h\vert\neq 0, \lambda=0, \widehat{\eta}=0$.
Therefore
\begin{eqnarray*}
&(x^\prime,0;\xi,\eta)\in \pi_{3*}\left(\Lambda_1+\Lambda_2 \right)\\
\Leftrightarrow & 
\left\lbrace 
\left( \begin{array}{ccc} x &;& \xi \\
h &;& 0 \end{array}\right)\in\Gamma\cup\underline{0}, (x,h)\in\text{supp }\psi\right\rbrace\\
\Leftrightarrow &  \pi_{3*}\left(\Lambda_1+\Lambda_2 \right)\cap T^*_IU\subset\Xi.
\end{eqnarray*}
It is immediate that $\pi_{3*}\Lambda_1=\emptyset$, finally
\begin{eqnarray*}
\left(\begin{array}{ccc}
x^\prime &; &\xi\\
0&; &\eta
\end{array} \right)\in \pi_{3*}\Lambda_2\cap T^*_IU
&\Leftrightarrow & \left\lbrace\begin{array}{c} \left\langle h,\eta \right\rangle=0\\
\xi=0\\
\lambda\eta=0\end{array} \right\rbrace
\text{ for } (x,\lambda h)=(x^\prime,0), (x,h)\in\text{supp }\psi\\
&\implies &\xi=0
\implies \pi_{3*}\Lambda_2\cap T^*_IU \subset N^*(I).
\end{eqnarray*}

Finally, we can summarize the bounds
that we obtained:
\begin{eqnarray}
\Lambda\cap T_I^*\mathbb{R}^{n+d}
&\subset & \Xi \cup N^*(I)
\end{eqnarray}
which establishes 
the claim
of our theorem.
\end{proof}

Now we prove that under
the assumptions 
of 
Theorem \ref{thm1},
the extension $\overline{t}$
constructed
is weakly homogeneous
of degree $s$ in
$\mathcal{D}^\prime_{\Gamma\cup N^*(I)\cup \Xi}(U)$.

\begin{thm}\label{thm2}
Let $s\in\mathbb{R}$ s.t. $s+d> 0$, 
$\Gamma\subset T^\bullet U$ 
a closed conic set stable 
by scaling.
Then 
the extension
$\overline{t}=\lim_{\varepsilon\rightarrow 0}t(1-\chi_{\varepsilon^{-1}})$ is in $E_s\left(\mathcal{D}^\prime_{\Xi\cup \Gamma\cup N^*(I)}(U)\right)$ for
$\Xi=\{(x,0;\xi,\eta) | (x,h)\in\text{supp }\psi,\,\ (x,h;\xi,0)\in\Gamma  \}$.

\end{thm}

\begin{proof}
For all test function $\varphi$, 
we study
the family
$\left(\left\langle\mu^{-s}\overline{t}_\mu,\varphi\right\rangle\right)_{\mu\in (0,1]}$. 
But since $\overline{t}=\lim_{\varepsilon\rightarrow 0}t(1-\chi_{\varepsilon^{-1}})$, it suffices
to study the family  $\mu^{-s}  \left(t(1-\chi_{\varepsilon^{-1}})\right)_{\varepsilon,\mu}$ for $\varepsilon\leqslant \mu$. 

A simple calculation using variable changes gives:
\begin{eqnarray*}
\forall 0 < \varepsilon\leqslant \mu\leqslant 1  , \mu^{-s} \left\langle \left(t(1-\chi_{\varepsilon^{-1}})\right)_\mu , \varphi \right\rangle
&=&\int_{\varepsilon}^1\frac{d\lambda}{\lambda} \mu^{-s-d} \left\langle t\psi_{\lambda^{-1}} , \varphi_{\mu^{-1}} \right\rangle  + \left\langle \mu^{-s}t_\mu(1-\chi_\mu),\varphi \right\rangle \\
&=&\int_{\varepsilon}^1\frac{d\lambda}{\lambda} \left(\frac{\lambda}{\mu}\right)^{s+d} \left\langle \lambda^{-s}t_\lambda\psi, \varphi_{\frac{\lambda}{\mu}} \right\rangle+
\left\langle \mu^{-s}t_\mu(1-\chi_\mu),\varphi \right\rangle\\
&=&\int_{\frac{\varepsilon}{\mu}}^{\frac{1}{\mu}}\frac{d\lambda}{\lambda} \lambda^{s+d} \left\langle (\lambda\mu)^{-s}t_{\lambda\mu}\psi , \varphi_\lambda \right\rangle  + 
\left\langle \mu^{-s}t_\mu(1-\chi_\mu),\varphi \right\rangle.
\end{eqnarray*}

First, note that the family
$(\mu^{-s}t_\mu)_{\mu\in(0,1]}$
is bounded in 
$\mathcal{D}^\prime_{\Gamma}(U\setminus I)$ and $(1-\chi_\mu)\rightarrow 0$
when $\mu\rightarrow 0$
therefore the family
$\left( \mu^{-s}t_\mu(1-\chi_\mu)\right)_{\mu\in(0,1]}$
is bounded.

The next thing we show is that the integral
$\int_{\frac{\varepsilon}{\mu}}^{\frac{1}{\mu}}\frac{d\lambda}{\lambda} \lambda^{s+d} \left\langle (\lambda\mu)^{-s}t_{\lambda\mu}\psi , \varphi_\lambda \right\rangle $ does not blow up
because its integrand vanishes when $\lambda$ is large enough.
Let $K$ be a compact subset of $\mathbb{R}^{n+d}$.
\begin{eqnarray*} 
\varphi\in\mathcal{D}_K(U)
&\implies &\exists R>0 \text{ s.t. }\text{supp }\varphi\subset
\{\vert h\vert\leqslant R\}\\ 
&\implies &
\text{supp }\varphi_\lambda\subset\{\vert h\vert\leqslant \lambda^{-1}R\}.
\end{eqnarray*}
Recall that 
$\pi$ was the projection 
$\pi:=(x,h)\in\mathbb{R}^{n+d}\mapsto 
x\in\mathbb{R}^n$. 
\begin{eqnarray*}
&&\pi\text{ is proper on }\text{supp }\psi\text{ and }
\pi(\text{supp }\varphi)\subset \mathbb{R}^n \text{ compact } \\
&\implies & 
\text{supp }t_{\lambda\mu}\psi|_{(K\times\mathbb{R}^d)\cap U}\subset\{a\leqslant \vert h\vert\leqslant b\}\text{ for  } 0<a<b\\
&\implies &\left\lbrace\lambda\geqslant \frac{R}{a}\implies\left\langle t_{\lambda\mu}\psi,\varphi_\lambda \right\rangle=0\right\rbrace\\
&\implies &\forall \mu\in(0,1],\varepsilon\leqslant\mu, \int_{\frac{\varepsilon}{\mu}}^{\frac{1}{\mu}}\frac{d\lambda}{\lambda} \lambda^{s+d} \left\langle (\lambda\mu)^{-s}t_{\lambda\mu}\psi , \varphi_\lambda \right\rangle=\int_{0}^{+\infty}\frac{d\lambda}{\lambda} 1_{\{\frac{\varepsilon}{\mu}\leqslant \frac{R}{a}\}}(\lambda)\lambda^{s+d} \left\langle(\lambda\mu)^{-s} t_{\lambda\mu}\psi ,\varphi_{\lambda} \right\rangle.
\end{eqnarray*}
For all $t$,
we define $t^\mu(x,h)=\mu^{-s} t(x,\mu h)$
and we consider
the family of distributions
$B=(t^\mu)_{\mu\in(0,1]}$
which is bounded in $E_s(\mathcal{D}^\prime_\Gamma(U))$.
Therefore
the result of
Lemma
(\ref{propfamily2})
implies that
the family 
$$\left(1_{\{\frac{\varepsilon}{\mu}\leqslant \frac{R}{a}\}}(\lambda)\lambda^{s+d-1}
\lambda^{-s}t^\mu(x,\lambda h)\psi(x,h)\right)_{0<\varepsilon\leqslant\mu\leqslant 1}$$
is \textbf{bounded} in
$\mathcal{D}^\prime_\Lambda([0,\frac{R}{a}]\times (U\setminus I)),$
for $\Lambda=\{ (\lambda,x,h;\tau,\xi,\eta) 
\in \dot{T}^*([0,\frac{R}{a}]\times (U\setminus I)) 
|(x,h)\in\text{supp }\psi, (x,h;\xi,\eta)\in \Gamma\cup\underline{0}\}$.
Therefore, we can repeat
the proof of Theorem
\ref{thm1} for the family
\begin{eqnarray}
\int_{\mathbb{R}\times\mathbb{R}^{n+d}}
d\lambda dxdh 1_{\{\frac{\varepsilon}{\mu}\leqslant \frac{R}{a}\}}(\lambda)\lambda^{s+d-1}
\lambda^{-s}t^\mu(x,\lambda h)\psi(x,h)\delta_\Delta(\Phi(\lambda,x,h),.)
\end{eqnarray}
Using the fact that
\begin{enumerate}
\item the H\"ormander product
of $ 1_{\{\frac{\varepsilon}{\mu}\leqslant \frac{R}{a}\}}(\lambda)\lambda^{s+d-1}
\lambda^{-s}t^\mu(x,\lambda h)\psi(x,h)$
with $\delta_\Delta(\Phi(\lambda,x,h),.)$ 
is 
hypocontinuous
\cite[Thm 7.1]{Viet-wf2}  
\item
the push--forward of
$ 1_{\{\frac{\varepsilon}{\mu}\leqslant \frac{R}{a}\}}(\lambda)\lambda^{s+d-1}
\lambda^{-s}t^\mu(x,\lambda h)\psi(x,h)
\delta_\Delta(\Phi(.),.)$
by the projection $\pi_3$ is continuous
in the normal topology hence bounded 
\cite[Thm 7.3]{Viet-wf2},
\end{enumerate}
we obtain the desired result.
\end{proof}

\subsubsection{Optimality of the wave front set
of the extension.}

We show with an example 
how our technique gives an 
\textbf{optimal bound}
for the wave front set
of the extension
of distributions
in a situation
where the
assumptions of the results 
of Brunetti--Fredenhagen~\cite[Lemma 6.1]{Brunetti2} are not
satisfied.
\paragraph{The wave front set of an example of extension
not handled by Brunetti--Fredenhagen's method.}
We work in $T^*\mathbb{R}^{3}$ with variables
$(x_1,x_2,h;\xi_1,\xi_2,\eta)$ 
and $I$ is the plane 
$\left(\mathbb{R}^2\times\{0\}\right)=\{h=0\}$.
Let $f\in C^\infty(\mathbb{R}\setminus \{0\})\cap L^\infty(\mathbb{R}), f>0$ 
which is nonsmooth at the origin 
and let us consider the function
$f(x_1)$ as a distribution
in the vector space $\mathbb{R}^3\setminus I $.
Then we prove the following
claim:
\begin{prop}
Let $\chi\in C^\infty(\mathbb{R})$ be a
smooth function s.t. $\chi(h)=1$ when $h\leq 1$
and $\chi(h)=0$ when $h\geqslant 2$. 
Then the 
family of distributions 
$f(x_1)(\chi(\varepsilon^{-1}h)-\chi(h))_{\varepsilon}$
converges to $f(x_1)$ when $\varepsilon\rightarrow 0$ 
in
$\mathcal{D}^\prime_V$ where 
$V=N^*(\{x_1=0\})\cup N^*I\cup\left(N^*(\{x_1=0\}) + N^*I\right)$.
\end{prop}
In fact, for all $\varepsilon>0$, the
wave front set of the distribution
$f(x_1)(\chi(\varepsilon^{-1}h)-\chi(h))$
is in $N^*(\{x_1=0\})$ therefore
it does not satisfy the assumption that
the closure of
$WF(f(x_1)(\chi(\varepsilon^{-1}h)-\chi(h)))$
should be contained in the
conormal $N^*(I)$ which is 
an important assumption of Theorem 6.9
in the paper \cite{Brunetti2} 
of Brunetti Fredenhagen.
\begin{proof}
Let $V$ be the smallest 
closed 
conic set such that
the family $f(x_1)(\chi(\varepsilon^{-1}h)-\chi(h)
)_{\varepsilon
\in(0,1]}$
is bounded in
$\mathcal{D}^\prime_V$.
It is obvious that
outside $\{h=0\}$ the cone
$V$ equals $N^*(\{x_1=0\})$.
It suffices to calculate
$V$ over $\{h=0\}$.

To estimate $V$ over $\{h=0\}$,
there are two cases to study: $x_1=0$ and $x_1\neq 0$ ($x_2$ is arbitrary).
We start with the case $x_1\neq 0$.
Let $\varphi$
be a test function:
\begin{eqnarray*}
&&\mathcal{F}\left(f(x_1)\varphi(x_1,x_2)(\chi(\varepsilon^{-1}h)-\chi(h)
) \right)\\
&=& \widehat{f\varphi}(\xi_1,\xi_2)\left(\widehat{\chi(\varepsilon^{-1}.)}(\eta) - \widehat{\chi}(\eta) \right) \\
&=& \widehat{f\varphi}(\xi_1,\xi_2)\left(\varepsilon\widehat{\chi}(\varepsilon\eta) - \widehat{\chi}(\eta) \right)
\end{eqnarray*}
Since $\widehat{\chi}\neq 0$ and is even analytic, we have
$\forall R>0, 
\sup_{\vert\eta\vert\geqslant R}\vert\widehat{\chi}(\eta)\vert=C(R)>0$.
This gives us the estimate
\begin{eqnarray*}
\forall R>0, \sup_{\eta}(1+\vert\eta\vert)^N 
\varepsilon\vert\widehat{\chi}(\varepsilon\eta)\vert&\geqslant &\left(1+\frac{R}{\varepsilon}\right)^N\varepsilon C(R)\\
&\geqslant & 
\varepsilon^{-N+1}R^NC(R)\rightarrow_{\varepsilon\rightarrow 0} \infty
\end{eqnarray*}
This implies that
$(\chi(\varepsilon^{-1}.)-\chi)$
is \textbf{bounded in} 
$\mathcal{D}^\prime_{N^*(I)}$, 
therefore using the fact that $f>0$, 
we find that $V$ corresponds
with the conormal $N^*(I)$
of $I$ as long as $f$ is smooth hence
outside $x_1=0$.

We conclude by studying the case where $x_1=0$.
Since $f$ is singular at $x_1=0$ and that 
$ss(f)=\{0\}=\pi_{T^*\mathbb{R}\mapsto \mathbb{R}}(WF(f))$, the wave front set of $f$ in the fiber $T_{0}^*\mathbb{R}$ over 
$x_1=0$ is non empty and we deduce there is a function $\varphi$ of the two variables
$(x_1,x_2)$
such that $\widehat{f\varphi}(\xi_1,\xi_2)$ 
has slow decrease in the direction $(\xi_1,0)$. 
The Fourier transform of $f\varphi(\chi(\varepsilon^{-1}.)-\chi)$
w.r.t $(x_1,x_2,h)$ equals $\widehat{f\varphi}(\xi_1,\xi_2)\left(\widehat{\chi(\varepsilon^{-1}.)}(\eta) - \widehat{\chi}(\eta) \right) $ from which
one easily
concludes that
$\{(0,x_2,0;\xi_1,0,\eta) \}=N^*(I)+N^*(\{x_1=0\})\subset V$.
\end{proof}

\paragraph{The singular case.}
In the next part,
we will deal with the singular case
where 
$-m-1<s+d\leqslant -m,m\in\mathbb{N}$.
Instead 
of calculating
the pairing
$\left\langle t(\chi-\chi_{\varepsilon^{-1}}), \varphi\right\rangle$,
we will subtract
from $\varphi$ its Taylor polynomial
in the $h$ variable
to a sufficient order,
therefore we
will pair
$t(\chi-\chi_{\varepsilon^{-1}})$
with the Taylor
remainder $I_m\varphi$
defined by
\begin{equation}\label{Taylorformula}
I_m\varphi(x,h)=\frac{1}{m!}\sum_{\vert\alpha\vert=m+1}h^\alpha\int_0^1(1-t)^m \left(\partial_h^\alpha \varphi\right)(x,th)dt.
\end{equation}
Then we will 
study the existence
of the limit:
\begin{eqnarray}\label{renormalizedformula}
\left\langle\overline{t},\varphi \right\rangle=\underset{\varepsilon\rightarrow 0}{\lim}
\left\langle t(\chi-\chi_{\varepsilon^{-1}}), I_m\varphi \right\rangle+\left\langle t(1-\chi),\varphi\right\rangle.
\end{eqnarray}

First, observe that
if the support of $\varphi$ does not meet
$I$, then $\varphi$ equals its Taylor remainder 
$I_m\varphi$ since $\varphi$ vanishes
at infinite order on the subspace $I$ and formula \ref{renormalizedformula}
is well defined and coincides with 
$\left\langle t,\varphi\right\rangle$.
Therefore if $\overline{t}$
were defined, it would 
be
an extension
of $t$.

\begin{thm}\label{thm3}
Let $s\in\mathbb{R}$ s.t. $-m-1<s+d\leqslant -m,m\in\mathbb{N}$, 
$U\subset \mathbb{R}^{n+d}$
a convex set,  
$\Gamma\subset T^\bullet U$ 
a closed conic set stable 
by scaling.
If $t\in\mathcal{D}^\prime(U\setminus I)$ is weakly
homogeneous of degree $s$ in 
$\mathcal{D}^\prime_\Gamma(U\setminus I)$,
then
formula (\ref{renormalizedformula})
defines an extension
$\overline{t}$ of $t$
and $WF(\overline{t})\subset 
WF(t)\cup N^*(I)\cup \Xi$
where 
\begin{equation*}
\Xi=\{(x,0;\xi,\eta) |(x,h)\in\text{supp }\psi,\,\ (x,h;\xi,0)\in\Gamma  \}.
\end{equation*}
\end{thm}

\begin{proof}
Before we state our theorem,
let us describe the
central new ingredient
of our proof. In appendix, we will
study the Schwartz kernel
of the
operator $I_m\in\mathcal{D}^\prime(\mathbb{R}^{n+d}\times\mathbb{R}^{n+d})$ 
realizing
the projection
on the 
Taylor remainder. We work in $\mathbb{R}^{n+d}\times\mathbb{R}^{n+d}$
with coordinates $(x,h,x^\prime,h^\prime)$ and we note
$I_m((x,h),(x^\prime,h^\prime))$ this Schwartz kernel.

\textbf{Step 1}
The distribution $I_m(.,.)$
plays the same role in the proof
of Theorem \ref{thm3}
as $\delta_\Delta(.,.)$ in the proof
of Theorem \ref{thm1}
and we prove in appendix ( Lemma
\ref{WFI_m}) 
that: 
\begin{eqnarray}\label{equationIm}
I_m(.,.)=\sum_{\vert\alpha\vert=m+1} h^\alpha R_\alpha(.,.)
\text{ where }\forall \alpha, R_\alpha(.,.)\in\mathcal{D}^\prime(\mathbb{R}^{n+d}\times \mathbb{R}^{n+d})\\
WF\left( R_\alpha(.,.)\right)\subset  \{(x,h,x,th;\xi,t\eta,-\xi,-\eta) |t\in[0,1], (\xi,\eta)\neq (0,0)\}.
\end{eqnarray}
\textbf{Step 2}
Let 
$\varphi\in\mathcal{D}(U)$ 
be a test 
function,
we have to establish 
the convergence
of the formula
\begin{eqnarray*}
\left\langle t(\chi-\chi_{\varepsilon^{-1}}), I_m(\varphi) \right\rangle+\left\langle t(1-\chi),\varphi\right\rangle
\end{eqnarray*} 
when $\varepsilon\rightarrow 0$.
As in the proof of
Theorem \ref{thm1},
we use the partition
of unity
to derive an equivalent
formula for 
$\underset{\varepsilon\rightarrow 0}{\lim}\left\langle t(\chi-\chi_{\varepsilon^{-1}}), I_m(\varphi) \right\rangle$
in terms of
the family
$\lambda^{-s}t_\lambda$:
\begin{eqnarray}\label{renormformula}
\underset{\varepsilon\rightarrow 0}{\lim}\left\langle t(\chi-\chi_{\varepsilon^{-1}}), I_m(\varphi) \right\rangle
&=&\int_0^1 \frac{d\lambda}{\lambda}\lambda^{s+d+m+1}\left\langle (\lambda^{-s}t_\lambda)\psi , \sum_{\vert\alpha\vert=m+1} h^\alpha R_{\alpha}(\varphi)_\lambda \right\rangle
\end{eqnarray} 
where $R_{\alpha}(\varphi)_\lambda(x,h)=\frac{1}{m!}\int_0^1(1-t)^m \partial_h^\alpha\varphi(x,t\lambda h)dt$
and the r.h.s. of (\ref{renormformula}) is absolutely convergent since
$s+d+m+1>0$.
It remains to prove that
the limit when $\varepsilon\rightarrow 0$ is a well defined
distribution.
The proof is similar to the proof in Theorem \ref{thm1}
except that we should use the fact that
the seminorms of $\psi\sum_{\vert\alpha\vert=m+1} h^\alpha R_{\alpha}(\varphi)_\lambda$
are controlled by the seminorms of $\varphi$ by Taylor's formula for the integral remainder.

\textbf{Step 3}
We are reduced to
prove the boundedness 
of the family distributions parametrized by $\varepsilon\in(0,1]$
\begin{eqnarray*}
t(\chi-\chi_{\varepsilon^{-1}})I_m(x^\prime,h^\prime)&=&
\int_{\mathbb{R}\times\mathbb{R}^{n+d}}
d\lambda dxdh 1_{[\varepsilon,1]}(\lambda)\lambda^{s+d+m}\lambda^{-s}t(x,\lambda h)\psi(x,h)\lambda^{-m-1}I_m((x,\lambda h),(x^\prime,h^\prime))
\end{eqnarray*}
in $\mathcal{D}^\prime_\Lambda$ where
$\Lambda=WF(t)\cup\Xi\cup N^*(I) $.

We can repeat exactly
the same proof 
as for Theorem
\ref{thm1} using 
parallel notations.
Set $B_{1,\varepsilon}(\lambda,x,h)=1_{[\varepsilon,1]}(\lambda)\lambda^{s+d+m}\lambda^{-s}t(x,\lambda h)\psi(x,h)$
then by Lemma \ref{propfamily2}, the family 
$\left(B_{1,\varepsilon}(\lambda,x,h)\otimes 1(x^\prime,h^\prime)\right)_{\varepsilon\in(0,1]}$
is bounded in $\mathcal{D}_{\Lambda_1}^\prime(U\times U\times\mathbb{R})$ where:
\begin{eqnarray*}
\Lambda_1=\{ \left(\begin{array}{ccc}
\lambda &;&\widehat{\lambda} \\
x &;& \widehat{\xi}\\
h &;&\widehat{\eta}\\
x^\prime &;& 0 \\
h^\prime &;& 0  
\end{array}\right)
| \left(\begin{array}{ccc}
x &;& \widehat{\xi}\\
h &;& \widehat{\eta}  
\end{array}\right)\in  \Gamma\cup \underline{0}, (x,h)\in\text{supp }\psi
\}\end{eqnarray*}
Equations \ref{equationIm} together with the
pull--back theorem of H\"ormander imply:
\begin{eqnarray*}
WF\left(\lambda^{-m-1}I_m(\Phi,.)\right)&\subset &\Lambda_2=\{\left(
\begin{array}{ccc}
\lambda &;& -\left\langle h,\eta \right\rangle\\
x &;& -\xi\\
h &;& -t\eta\\
x^\prime &;&\xi\\
h^\prime &;&\eta
\end{array}
\right) 
| (x,t h)=(x^\prime,h^\prime), t\in[0,\lambda], (\xi,\eta)\neq (0,0)  \}.
\end{eqnarray*}
Note that $\Lambda_1
\cap -\Lambda_2=\emptyset$ 
implies that the family
of products
$\left(B_{1,\varepsilon}(\lambda,x,h)\lambda^{-m-1}I_m((x,\lambda h),(x^\prime,h^\prime))\right)_{\varepsilon}$
is bounded in $\mathcal{D}^\prime_{\Lambda_1+\Lambda_2\cup\Lambda_1\cup\Lambda_2}$
by hypocontinuity of the H\"ormander product \cite[Thm 7.1]{Viet-wf2}.
As in the proof of proposition \ref{thm1}, 
we have $
t(\chi-\chi_{\varepsilon^{-1}})I_m=\pi_3(B_{1,\varepsilon}\lambda^{-m-1}I_m(\Phi,(x^\prime,h^\prime)))$
therefore, in order to conclude, it suffices to control
the family
$\left(\pi_{3*}(B_{1,\varepsilon}\lambda^{-m-1}I_m(\Phi,(x^\prime,h^\prime)))\right)_\varepsilon$
in $\mathcal{D}^\prime_\Lambda$ where
$\Lambda=WF(t)\cup\Xi\cup N^*(I) $,
using continuity of the push--forward 
\cite[Theorem 7.3]{Viet-wf2},
we have the following
estimate:
\begin{eqnarray}\label{formula3}
\left(\Lambda\cap T^*_IU\right)
\subset \pi_{3*}\left(\Lambda_1+\Lambda_2\cup\Lambda_1\cup\Lambda_2 \right)\cap T^*_IU .
\end{eqnarray}
We study the closed conic set  
$\pi_{3*}\left(\Lambda_1+\Lambda_2 \right)\cap T^*_IU$:
\begin{eqnarray*}
&(x^\prime,0;\xi,\eta)\in \pi_{3*}\left(\Lambda_1+\Lambda_2 \right)\\
\Leftrightarrow & 
\left\lbrace 
\begin{array}{c}
\widehat{\lambda}-\left\langle h,\eta \right\rangle=0\\
\widehat{\xi}-\xi=0\\
\widehat{\eta}-t\eta=0
\end{array}
| \exists t\in[0,\lambda] \text{ s.t. }
(x,t h)=(x^\prime,0), (x,h)\in\text{supp }\psi,
\left( \begin{array}{ccc} x &;& \widehat{\xi} \\
h &;& \widehat{\eta} \end{array}\right)\in\Gamma\cup\underline{0}\right\rbrace
\end{eqnarray*}
has a solution. Note that
$\left\lbrace\begin{array}{c}(x,h)\in\text{supp }\psi \\ 
(x,t h)=(x^\prime,0)\\
\widehat{\eta}-t\eta=0
\end{array}\right\rbrace\implies \vert h\vert\neq 0, t=0, \widehat{\eta}=0$.
Therefore
\begin{eqnarray*}
&(x^\prime,0;\xi,\eta)\in \pi_{3*}\left(\Lambda_1+\Lambda_2 \right)\\
\Leftrightarrow & 
\left\lbrace 
\left( \begin{array}{ccc} x &;& \xi \\
h &;& 0 \end{array}\right)\in\Gamma\cup\underline{0}, (x,h)\in\text{supp }\psi\right\rbrace\\
\Leftrightarrow &  \pi_{3*}\left(\Lambda_1+\Lambda_2 \right)\cap T^*_IU\subset\Xi.
\end{eqnarray*}
It is immediate that $\pi_{3*}\Lambda_1=\emptyset$, finally
\begin{eqnarray*}
\left(\begin{array}{ccc}
x^\prime &; &\xi\\
0&; &\eta
\end{array} \right)\in \pi_{3*}\Lambda_2\cap T^*_IU
&\Leftrightarrow & \left\lbrace\begin{array}{c} \left\langle h,\eta \right\rangle=0\\
\xi=0\\
t\eta=0\end{array} \right\rbrace
\text{ for } (x,t h)=(x^\prime,0), t\in[0,\lambda], (x,h)\in\text{supp }\psi\\
&\implies &\xi=0
\implies \pi_{3*}\Lambda_2\cap T^*_IU \subset N^*(I).
\end{eqnarray*}
Finally, we can summarize the bounds
that we obtained:
\begin{eqnarray}
\Lambda\cap T_I^*\mathbb{R}^{n+d}
&\subset & \Xi \cup N^*(I)
\end{eqnarray}
which establishes 
the claim
of our theorem.
\end{proof}

We want to show that our extension
is weakly homogeneous
in $\mathcal{D}^\prime_\Gamma$.

\begin{prop}\label{scal2}
Under the assumptions of proposition (\ref{thm3}),
if $s$ is not an integer then 
the extension map $t\in E_{s}\left(\mathcal{D}^\prime_\Gamma\left(U\setminus I\right)\right)
\longmapsto \overline{t}\in E_{s}(\mathcal{D}^\prime_{\Gamma\cup N^*(I)\cup\Xi}(U))$
is bounded.
\end{prop}
\begin{prop}\label{scal3noninteger}
Under the assumptions of proposition (\ref{thm3}), 
if $s+d$ 
is a $\textbf{non positive integer}$ 
then
\begin{itemize}
\item the extension map $t\in E_{s}\left(\mathcal{D}^\prime_\Gamma\left(U\setminus I\right)\right)
\longmapsto \overline{t}\in E_{s^\prime}(\mathcal{D}^\prime_{\Gamma\cup N^*(I)\cup\Xi}(U)),\forall s^\prime<s$ 
is bounded,
\item the family of 
distributions
$\lambda^{-s}(\log\lambda)^{-1}\overline{t}_\lambda$ 
is bounded
in $\mathcal{D}^\prime_{\Gamma\cup N^*(I)\cup\Xi}(U)$.
\end{itemize}
\end{prop}

\begin{proof}
Choose
a test function 
$\varphi$.
To check the homogeneity 
of the renormalized 
integral is a 
little tricky since 
we have to take the 
scaling of counterterms 
into account.
When we scale the test
function $\varphi$ 
then we should scale simultaneously 
the Taylor polynomial 
$(P_m\varphi)_\lambda$ 
and the remainder 
$(I_m\varphi)_\lambda$: 
$$\varphi_\lambda=
(P_m\varphi)_\lambda+(I_m\varphi)_\lambda
=P_m\varphi_\lambda
+I_m\varphi_\lambda.$$  
We want to know to which scale space 
$E_{s^\prime}(\mathcal{D}^\prime_{\Gamma\cup N^*(I)})$ 
the distribution $\overline{t}$ 
belongs:
$$\mu^{-s^\prime}\left\langle \overline{t}_\mu , \varphi \right\rangle
=\mu^{-s^\prime}\int_0^1\frac{d\lambda}{\lambda} \mu^{-d}\left\langle t\psi_{\lambda^{-1}}, (I_m\varphi)_{\mu^{-1}} \right\rangle $$
$$= \mu^{-s^\prime}\int_0^1\frac{d\lambda}{\lambda} \lambda^d\mu^{-d}\left\langle t_\lambda\psi, (I_m\varphi)_{\lambda\mu^{-1}} \right\rangle.$$ 
For the moment, we find that:
$$\mu^{-s^\prime}\left\langle \overline{t}_\mu ,
\varphi \right\rangle
=\mu^{s-s^\prime}\int_0^1\frac{d\lambda}{\lambda} 
\left(\frac{\lambda}{\mu}\right)^{s+d}
\left\langle \left(\lambda^{-s}t_\lambda\right)\psi, 
(I_m\varphi)_{\frac{\lambda}{\mu}} \right\rangle.$$ 
The test function $\varphi$ 
is supported in $\{\vert h\vert\leqslant R\}$
therefore $\varphi_{\frac{\lambda}{\mu}}$ is supported
on $\vert h\vert\leqslant \frac{\mu R}{\lambda}$
thus when $\frac{R\mu}{\lambda}\leqslant a\Leftrightarrow \frac{R\mu}{a}\leqslant \lambda$, 
the support of $\varphi_{\frac{\lambda}{\mu}}$
does not meet 
the support
of $\lambda^{-s}t_\lambda\psi$
because $\psi$
is supported on $a\leqslant\vert h\vert$,
whereas the polynomial part $P_m\varphi$ 
is supported 
everywhere since it is a Taylor polynomial. 
Consequently, we must split the scaled
distribution
$\mu^{-s}\overline{t}_\mu=I^\mu_1+I^\mu_2$
in two parts, where:
$$\left\langle I^\mu_1, \varphi\right\rangle=
\int_0^{\frac{R\mu}{a}}\frac{d\lambda}{\lambda}
\left(\frac{\lambda}{\mu}\right)^{s+d}
\left\langle \left(\lambda^{-s}t_\lambda\right)\psi, 
(I_m\varphi)_{\frac{\lambda}{\mu}} \right\rangle $$ 
$$=\int_0^{\frac{R\mu}{a}}\frac{d\lambda}{\lambda} \left(\frac{\lambda}{\mu}\right)^{(d+s+m+1)}\left\langle\left(\lambda^{-s}t_\lambda\right)\psi, \sum_{\vert \alpha\vert =m+1} h^\alpha R_{\alpha,\frac{\lambda}{\mu}}\right\rangle .$$ 
$$\left\langle I^\mu_2,\varphi\right\rangle
= \int_{\frac{R\mu}{a}}^1
\frac{d\lambda}{\lambda}
\left(\frac{\lambda}{\mu}\right)^{s+d}\underset{\text{
no contribution of $\varphi_{\frac{\lambda}{\mu}}$ since }
\frac{R\mu}{a}\leqslant \lambda}{\left\langle 
\lambda^{-s}t_\lambda , \varphi_{\frac{\lambda}{\mu}}
-(P_m\varphi)_{\frac{\lambda}{\mu}} \right\rangle}.$$ 
We make a simple variable change for $I^\mu_1$: 
$$\left\langle I^\mu_1,\varphi\right\rangle
=\int_0^{\frac{R}{a}}\frac{d\lambda}{\lambda} 
\lambda^{(d+s+m+1)}\left\langle(\lambda\mu)^{-s}t_{\lambda\mu}\psi, 
\sum_{\vert \alpha\vert =m+1} h^\alpha R_{\alpha,\lambda}\right\rangle\hfill$$ 
then following
the proof
of proposition
\ref{thm1},
we note that
\begin{equation}
I^\mu_1=\int_{\mathbb{R}\times\mathbb{R}^{n+d}}
d\lambda dxdh \lambda^{s+d+m}
\lambda^{-s}
\Phi^*(\mu^{-s}t_\mu)(\lambda,x,h)
1_{[0,\frac{R}{a}]}\psi(x,h)\lambda^{-m-1}I_m(\Phi(.),.).
\end{equation}
Therefore,
we can repeat
the proof of
proposition \ref{thm3}
for the bounded family
$(\mu^{-s}t_\mu)_\mu$
in $\mathcal{D}^\prime_\Gamma(U)$ and we deduce
that $(I^\mu_1)_\mu$ is bounded in 
$\mathcal{D}^\prime_\Lambda $ where $\Lambda=\Gamma\cup N^*(I)\cup\Xi$.

Notice that in the second term $I^\mu_2$ 
only the counterterm $P_m\varphi$ 
contributes 
$$I^\mu_2=\int_{\frac{R\mu}{a}}^1\frac{d\lambda}{\lambda}
\left(\frac{\lambda}{\mu}\right)^{s+d}\left\langle \lambda^{-s}t_\lambda\psi ,-(P_m\varphi)_{\frac{\lambda}{\mu}} \right\rangle $$
$$=\int_{\frac{R\mu}{a}}^1\frac{d\lambda}{\lambda}\left\langle  \lambda^{-s}t_\lambda\psi ,-\sum_{\vert \alpha\vert \leqslant m} \left(\frac{\lambda}{\mu}\right)^{s+d+\vert\alpha\vert}\frac{h^\alpha}{\alpha!}\pi^* \left(i^*\partial_h^\alpha \varphi\right) \right\rangle.$$ 
We reformulate $I^\mu_2$ as
\begin{equation}
I^\mu_2=-\int_{\mathbb{R}\times\mathbb{R}^{n+d}}
\frac{d\lambda}{\lambda} dxdh 
\lambda^{-s}
\Phi^*t(\lambda,x,h)
1_{[\frac{R\mu}{a},1]}\psi(x,h)\sum_{\vert \alpha\vert \leqslant m} \left(\frac{\lambda}{\mu}\right)^{s+d+\vert\alpha\vert}\frac{h^\alpha}{\alpha!}\pi^* \left(i^*\partial_h^\alpha 
\delta_\Delta(.,.)\right)
\end{equation}

Then notice that by assumption $s+d\leqslant -m$ and 
$\vert \alpha\vert$ ranges from $0$ to $m$ 
which implies that we always have 
$s+d+\vert \alpha\vert\leqslant 0$. 

If $s+d<m$ then
for all
$\alpha$ such that
$0\leqslant \vert\alpha\vert\leqslant m$ we have the inequality
$s+d+\vert\alpha\vert<0$,
hence
the family
of functions 
$1_{[\frac{R\mu}{a},1]}
\left(\frac{\lambda}{\mu}\right)^{s+d+\vert\alpha\vert}\lambda^{-1}$
is integrable w.r.t the variable $\lambda$
uniformly in the parameter $\mu$
since: 
\begin{equation*}
\Vert 1_{[\frac{R\mu}{a},1]}\left(\frac{\lambda}{\mu}\right)^{s+d+\vert\alpha\vert}
\lambda^{-1}
\Vert_{L^1(\mathbb{R})}=\frac{1}{\vert s+d+\vert\alpha\vert\vert}\underset{\text{no blow up when }\mu\rightarrow 0}{\underbrace{\left|\left(\frac{1}{\mu}\right)^{s+d+\vert\alpha\vert}-\left(\frac{R}{a}\right)^{s+d+\vert\alpha\vert}\right|}} .
\end{equation*}
Therefore the family $(\lambda^{-s-1}
\Phi^*t(\lambda,x,h)
1_{[\frac{R\mu}{a},1]}\psi(x,h)
\left(\frac{\lambda}{\mu}\right)^{s+d+\vert\alpha\vert}
\frac{h^\alpha}{\alpha!})_{\mu\in (0,1)}$
is bounded in $\mathcal{D}^\prime_V$
where $V=\{(\lambda,x,h;\tau,\xi,\eta)\in \dot{T}^*(\mathbb{R}\times \mathbb{R}^{n+d});\lambda\in[0,1],(x,h;\xi,\eta)\in\Gamma\cup\underline{0}\}$ 
by Proposition \ref{propfamily2} 
and we can repeat
the proof of proposition
\ref{thm3} where the Schwartz kernel
$I_m(\Phi(.),.)$ should be
replaced with the distribution
$\pi^* \left(i^*\partial_h^\alpha 
\delta_\Delta(.,.)\right) $ whose wave front set
is calculated in Lemma \ref{wfPm} in appendix, 
the proof of Proposition
\ref{thm3} still applies in our case
since 
$WF\left(\pi^* \left(i^*\partial_h^\alpha 
\delta_\Delta(.,.)\right)\right)\subset 
WF\left(I_m(\Phi(.),.)\right)$.
However if $s+d+m=0$ then for
$\vert\alpha\vert=m$, we find that
the family of functions
$$\left(1_{[\frac{R\mu}{a},1]}
\left(\frac{\lambda}{\mu}\right)^{s+d+\vert\alpha\vert}\lambda^{-1}
=\lambda^{-1}1_{[\frac{R\mu}{a},1]}\right)_{\mu\in(0,1]}$$
is no longer bounded in
the $L_\lambda^1([0,1])$ for $\mu\in(0,1]$
but exhibits a 
\textbf{logarithmic divergence}:
\begin{equation*}
\forall\mu\in(0,1],  \Vert 1_{[\frac{R\mu}{a},1]}
\lambda^{-1}\Vert_{L^1(\mathbb{R})}=\log(\frac{R\mu}{a})\leqslant
\log\mu+\log(\frac{R}{a}).
\end{equation*}
Then it is easy
to conclude that
$(\log\lambda)^{-1}\lambda^{-s}\overline{t}_\lambda$
is bounded in $\mathcal{D}^\prime_{\Xi\cup N^*(I)\cup \Gamma}(U)$.
\end{proof}
\subsection{The general extension in the flat case.}
\label{genexttheoremssubsection}
For the sequel, we recall
that $\chi\in C^\infty(\mathbb{R}^{n+d})$
is our partition
of unity used to construct the extension
and $\psi=-h\frac{d\chi}{dh}$.
\begin{thm}\label{flatgenextensionthm}
Let $s\in\mathbb{R}$,  
$\Gamma\subset T^\bullet U$ 
a closed conic set stable 
by scaling.
If $t\in\mathcal{D}^\prime(U\setminus I)$ is weakly
homogeneous of degree $s$ in 
$\mathcal{D}^\prime_\Gamma(U\setminus I)$,
then
\begin{enumerate}
\item there is an extension
$\overline{t}\in\mathcal{D}^\prime(U)$ 
of $t$
where:
\begin{eqnarray*}
WF(\overline{t})\subset 
WF(t)\cup N^*(I)\cup \Xi, \,\ 
\Xi=\{(x,0;\xi,\eta) |(x,h)\in\text{supp }\psi, (x,h;\xi,0)\in\Gamma  \}.
\end{eqnarray*}
\item $\overline{t}$ is in $E_{s,\Gamma\cup \Xi\cup N^*(I)}(U)$
if $-s-d\notin \mathbb{N}$ and $\overline{t}\in E_{s^\prime,\Gamma\cup \Xi\cup N^*(I)}(U),s^\prime<s$
otherwise.
\end{enumerate}
\end{thm}
We give here the proof
of an important particular case of the above theorem:
\begin{thm}
Under the assumptions of the above theorem
if $\left(\overline{\Gamma}\cap T^\bullet_IU\right)\subset
N^*(I)$ then 
\begin{enumerate}
\item there is an extension
$\overline{t}\in\mathcal{D}^\prime(U)$ 
of $t$
where:
\begin{eqnarray*}
WF(\overline{t})\subset 
WF(t)\cup N^*(I).
\end{eqnarray*}
\item $\overline{t}$ is in $E_{s,\Gamma\cup N^*(I)}(U)$
if $-s-d\notin \mathbb{N}$ and $\overline{t}\in E_{s^\prime,\Gamma\cup N^*(I)}(U),s^\prime<s$
otherwise.
\end{enumerate}
\end{thm}
\begin{proof}
The proof proceeds in two steps. First, we show that
there exists a neighborhood
$V$ of $I=\{h=0\}$ such that
$\forall (x,h;\xi,\eta)\in T^\bullet V\cap \Gamma$,
$\eta\neq 0$. 
In the second part, we explain that
by carefully choosing $\chi$
in such a way that $\text{supp }\chi\subset V$,
the subset $\Xi$ will be empty.

\textbf{Step 1}, we prove that for all
compact set $K$ there is some neighborhood $V$
of $I$ such that $\Gamma\cap T^\bullet_{K\cap V} U$
does not meet the set $\{(x,h;\xi,0)| \xi\neq 0\}$.
Then it follows immediately
by a covering argument that there exists a neighborhood
$V$ of $I=\{h=0\}$ such that
$\forall (x,h;\xi,\eta)\in T^\bullet V\cap \Gamma$,
$\eta\neq 0$.
By contradiction assume there is some
compact set $K$ such that for all
$V_n=\{\vert h\vert\leqslant n^{-1} \}$, there is some
$(x_n,h_n;\xi_n,0)\in T^\bullet_{K\cap V_n} U\cap \Gamma $.
By extracting a convergent subsequence one easily
concludes that there would be a sequence
$(x_n,h_n;\xi_n,0)\rightarrow (x,0;\xi,0)\in \Gamma$, 
contradiction !

\textbf{Step 2} 
We choose a function $\chi$ which
equals $1$ in some neighborhood
of $I$ and $\chi$
is supported in $V$.
Therefore the function
$\psi=-\rho\chi$ is supported
in $V$.
But the set $\Gamma \cap T_V^*U$ does not meet
the set $\{(x,h;\xi,0)| \xi\neq 0\}$
therefore 
the set $\Xi=\{(x,0;\xi,\eta)| (x,h)\in\text{supp }\psi 
(x,h;\xi,0)\in \Gamma  \}$ is empty and the
conclusion follows.
\end{proof}

\section{The extension theorem
for $E_{s,I}$.}
\label{extthmsection}
We are now ready to prove Theorem \ref{extthmEsintro2} and 
some part
of the claim of Theorem \ref{extthmEsintro1}:
\begin{thm}\label{thmfi}
Let $U$ be an open neighborhood of $I\subset M$, if $t\in E_{s,I}(U\setminus I)$ then there exists
an extension $\overline{t}$ in $E_{s^\prime,I}(U)$ where $s^\prime=s$ if $-s-d\notin\mathbb{N}$ 
and
$s^\prime<s$ otherwise.
\end{thm}
\begin{proof}
$t\in E_{s,I}(U\setminus I)$ implies that
for all $p\in  I$, there is
some open chart $\psi:V_p\subset U\mapsto \mathbb{R}^{n+d}$, 
$\psi(I)\subset \mathbb{R}^n\times\{0\}$
where $\lambda^{-s}(\psi_*t)(x,\lambda h) $
is bounded in $\mathcal{D}^\prime(\psi\left(V_p\setminus I\right))$.
Moreover, we must choose $V_p$ in such a way that
its image $U=\psi(V_p)\subset \mathbb{R}^{n+d}$
is of the form $U_1\times U_2$ where $U_1\subset\mathbb{R}^n, U_2\subset\mathbb{R}^d$
and $\lambda U_2\subset U_2,\forall\lambda\in[0,1]$.
 $\cup_{p\in I} V_p$ forms an open cover
of $I$, consider a locally finite 
subcover
$\cup_{a\in A} V_a$ and denote by $(\psi_a)_{a\in A}$
the corresponding charts.
For every $a\in A$, Theorem \ref{flatgenextensionthm} yields an extension
$\overline{\psi_{a*}t}$ of $\psi_{a*}t$ in $E_{s^\prime,I}(\psi_a(V_a))$
and by diffeomorphism invariance of $E_{s^\prime,I}$ (Theorem \ref{Esdiffinvariance}), 
the element $\psi_a^*\overline{\psi_{a*}t}$
belongs to $E_{s^\prime,I}(V_a)$. Choose a partition of unity
$(\varphi_a)_a$ subordinated to the open cover
$\cup_{a\in A} V_a$, then an extension 
of $t$ reads $\sum_{a\in A} \varphi_a  \psi_a^*\overline{\psi_{a*}t}
+\left(1-\sum_{a\in A} \varphi_a\right)t$ and belongs
to $E_{s^\prime,I}(U)$ by the gluing property for $E_{s^\prime,I}$.
\end{proof}
\subsection{A converse result.}
Before we move on, let us prove a  converse theorem, namely that given any distribution $t\in \mathcal{D}^\prime\left(\mathbb{R}^{n+d}\right)$, for all
relatively compact subset $U$, 
we can find $s_0\in \mathbb{R}$ such that for all $s\leqslant s_0$, $t\in E_{s,I}(U)$, this means
morally
that any distribution
has ``finite
scaling degree''
along an arbitrary
vector subspace. 
We also have the property that $\forall s_1\leqslant s_2, t\in E_{s_2,I}\implies t\in E_{s_1,I}$. This means that the spaces $E_{s,I}$ are $\textbf{filtered}$. 
We work in $\mathbb{R}^{n+d}$ where $I=\mathbb{R}^n\times\{0\}$ and $\rho=h^j\frac{\partial}{\partial h^j}$:
\begin{prop}\label{converse}
Let $U$ be a relatively compact 
convex open set
and $t\in \mathcal{D}^\prime(\mathbb{R}^{n+d})$.
If
$t$ is of order $k$ on $U$,
then $t\in E_{s,I}(U)$
for all
$s\leqslant d+k$, 
where $d$ is the \textbf{codimension}
of $I\subset\mathbb{R}^{n+d}$. 
In particular
any compactly 
supported
distribution
is in
$E_{s,I}(\mathbb{R}^{n+d})$
for some $s$.
\end{prop}
\begin{proof}
First notice if a function $\varphi\in \mathcal{D}(U)$, then the family of scaled functions $(\varphi_{\lambda^{-1}})_{\lambda\in(0,1]}$ has support contained in a compact set $K=\{(x,\lambda h) \vert (x,h)\in\text{supp }\varphi,\lambda\in (0,1] \}$.
We recall that for any distribution $t$, there exists $k,C_K$ such that
$$\forall\varphi \in \mathcal{D}_K(U), \vert\left\langle t,\varphi \right\rangle\vert \leqslant C_K \pi_{K,k}(\varphi).$$ 
$$\vert\left\langle t_{\lambda},\varphi \right\rangle\vert=\vert\lambda^{-d}\left\langle t,\varphi_{\lambda^{-1}} \right\rangle\vert\leqslant C_K \lambda^{-d}\pi_{K,k}(\varphi_{\lambda^{-1}})\leqslant C_K \lambda^{-d-k}\pi_{K,k}(\varphi).$$
So we find that $\lambda^{d+k}\left\langle t_\lambda,\varphi \right\rangle $ is bounded which yields the conclusion.
\end{proof}
Then Theorem \ref{extthmEsintro1} follows from
Proposition \ref{converse} and the diffeomorphism invariance
of $E_{s,I}$.

\section{The subspace $E_{s,N^*(I)}(U)$.}
\label{Esmicrolocsection}
It is a
central assumption
of our extension theorems
that the family 
$\left(\lambda^{-s}t_\lambda\right)_\lambda$ is
bounded in $\mathcal{D}^\prime_{\Gamma}$
and we found that
in the particular case
where $\overline{\Gamma}|_I\subset N^*(I)$
then
the wave front set
of the extension
is minimal i.e.
\begin{eqnarray}
WF(\overline{t})\subset WF(t)\cup N^*(I).
\end{eqnarray}

In this section, we
generalize the previous situation
to manifolds. We define
a subspace $E_{s,N^*(I)}$
of $E_s$
which contains distributions $t$ 
such that their extension
$\overline{t}$ satisfies
$WF(\overline{t})\subset WF(t)\cup N^*(I)$.

\subsection{The conormal landing condition.}
\label{conormlandingsubsec}
\begin{defi}
Let $U$ be an open neighborhood  
of $I$.
A closed conic set $\Gamma\subset T^\bullet \left(U\setminus I\right)$
(resp $\Gamma\subset T^\bullet U$) 
is said to satisfy the
\emph{conormal landing
condition} if
$(\overline{\Gamma}\cap T^\bullet_IU)\subset N^*I$ (resp $\left(\Gamma\cap T^\bullet_IU\right)\subset N^*I$) where
$\overline{\Gamma}$ is the closure of $\Gamma$ in $T^\bullet U$.
\end{defi}
The \emph{conormal landing
condition} which concerns the closure
of $\Gamma$ over $T^*_IU$ is clearly intrinsic
and does not depend on chosen coordinates.
The following is a stability
result
for sets
which satisfy the 
conormal landing condition.
\begin{lemm}\label{stabconormlanding}
Let $U$ be some open neighborhood 
of $I$, $\Gamma\subset T^\bullet \left(U\setminus I\right)$,
and
$\Phi\in C^\infty([0,1]\times U,U)$ be such that
$\Phi(\lambda,.)$ is a germ of diffeomorphism along $I$, $\Phi(\lambda,.)|_I$
is the identity map
for all $\lambda\in(0,1]$ and 
\begin{eqnarray*}
\forall (x,h;\xi,\eta)\in N^*(I), (\Phi_\lambda^{-1}(x,h);(\xi,\eta)\circ d\Phi_\lambda)= (x,h;\xi,\eta).
\end{eqnarray*}
If 
$\Gamma$
satisfies the conormal
landing condition then
the cone $\Gamma^\prime$ defined as
\begin{equation}
\Gamma^\prime=\underset{\lambda\in(0,1]}{\cup} \Phi(\lambda)^*\Gamma
\end{equation}
also does.
\end{lemm}
In the terminology of Lemma \ref{lemmlift}
in appendix, the condition of the above Lemma means
that the cotangent lift $T^*\Phi(\lambda,.)$
restricted to $N^*(I)$ acts as the identity map.
\begin{proof}
%
Let $(x,0;\xi,\eta)$ be in the closure
of $\Gamma^\prime$,
then there exists a sequence 
$(\lambda_n,x_n,h_n;\xi_n,\eta_n)_n$ such that
$(\Phi_{\lambda_n}^{-1}(x_n,h_n);(\xi_n,\eta_n)\circ d\Phi_{\lambda_n})\rightarrow  (x,0;\xi,\eta) $.
By compactness of $[0,1]$, 
we can always extract a subsequence so that
$\lambda_n\rightarrow \lambda_0\in[0,1]$. Then 
necessarily $(\Phi_{\lambda_0}^{-1}(x_n,h_n);(\xi_n,\eta_n)\circ d\Phi_{\lambda_0})\rightarrow  (x,0;\xi,\eta) $
which implies that $(x_n,h_n;\xi_n,\eta_n)\rightarrow (\Phi_{\lambda_0}(x,0);(\xi,\eta)\circ d\Phi^{-1}_{\lambda_0})=(x,0;\xi,\eta)$ 
since the cotangent lift $T^*\Phi(\lambda_0,.)|_{N^*(I)}$ is the identity map and 
$T^*\Phi(\lambda_0,.)|$ is a diffeomorphism.
\end{proof}

%

\subsection{Construction of $E_{s,N^*(I)}$.}
We keep the notations of the above subsection.
We give a preliminary definition
of the space $E^\rho_{s}(\mathcal{D}_\Gamma^\prime(U))$
for $\rho$-convex open sets $U$ and a given closed cone $\Gamma\subset T^\bullet U$ 
which depends
on the choice of $\rho$. 
\begin{defi}
Let $U$ be $\rho$ convex set, $\Gamma\subset T^\bullet U$ a closed cone, 
then 
$E^\rho_{s}(\mathcal{D}_\Gamma^\prime(U))$ is defined
as the space of distribution $t$
such that the family 
$\left(\lambda^{-s}e^{\log\lambda\rho*}t\right)_{\lambda\in(0,1]}$
is bounded in $\mathcal{D}^\prime_\Gamma(U)$.
\end{defi}
We next define a localized
version of the
above space around an element $p\in I$.
\begin{defi}
$t$ belongs to $E^\rho_{s,N^*(I),p}$ 
if 
there exists a $\rho$-convex open set $U$ s.t. $\overline{U}$
is a neighborhood of $p$ 
and $t\in E_s^\rho(\mathcal{D}^\prime_\Gamma(U))$
for some $\Gamma\subset T^\bullet U$
which satisfies the conormal 
landing condition.
\end{defi}

\begin{thm}\label{thminvmuloc}
Let $t\in \mathcal{D}^\prime(M\setminus I)$ and $p\in I$. If $t$ belongs to $E^\rho_{s,N^*(I),p}$  
for some Euler vector field $\rho$ then it is so for any Euler vector field.
\end{thm}
\begin{proof}
Let $\rho_1,\rho_2$ 
be two Euler vector fields and
$t$ belongs to $E^{\rho_1}_{s,N^*(I),p}$. It suffices
to establish that the family $\left(\lambda^{-s}e^{\log\lambda\rho_2*}t\right)_\lambda  $
is bounded in $\mathcal{D}^\prime_{\Gamma_2}(V^\prime_p\setminus I)$
for some neighborhood $V^\prime_p$ of $p$ and $\Gamma_2$ satisfying 
the conormal landing condition.
We use Proposition \ref{propositionvariablefamily} which states that locally 
there exists a smooth family of germs of diffeomorphisms 
$\Phi(\lambda): V_p\mapsto M$ such that 
$\forall\lambda\in[0,1],\Phi(\lambda)(p)=p$ 
and 
$\Phi(\lambda)$ relates the two scalings: 
$$e^{\log\lambda\rho_2*}=\Phi(\lambda)^*e^{\log\lambda\rho_1*} .$$

Assume that $V_p$ is chosen small enough
so that 
$\lambda^{-s} e^{\log\lambda\rho_1*}t$ is bounded in $\mathcal{D}^\prime_{\Gamma_1}(V_p\setminus I)$, 
then by \cite[Theorem 6.9]{Viet-wf2}, 
we deduce that the family
\begin{eqnarray*}
\left(\Phi(\lambda)^*\left(\lambda^{-s}e^{\log\lambda\rho_1*}t\right)\right)_\lambda=\left(\lambda^{-s}e^{\log\lambda\rho_2*}t\right)_\lambda 
\end{eqnarray*}
is in fact bounded in $\mathcal{D}^\prime_{\Gamma_2}(V^\prime_p\setminus I)$ for some
smaller neighborhood $V^\prime_p$ of $p$ and 
with $\Gamma_2$ given by the equation
$$ \Gamma_2= \bigcup_{\lambda\in[0,1]}\Phi(\lambda)^* \Gamma_1 .$$
By Lemma \ref{Idlift} proved in appendix, 
the family $\Phi(\lambda)$ satisfies:
\begin{eqnarray*}
\forall (x,h;\xi,\eta)\in N^*(I), (\Phi_\lambda^{-1}(x,y);(\xi,\eta)\circ d\Phi_\lambda)= (x,h;\xi,\eta).
\end{eqnarray*}
which implies by Lemma \ref{stabconormlanding} that $\Gamma_2$
satisfies the conormal landing condition concluding our proof.
\end{proof}
The previous theorem allows us 
to define spaces $E_{s,N^*(I),p}, E_{s,N^*(I)}$ which makes no mention
of the choice of Euler vector field $\rho$:
\begin{defi}\label{defEsmuloc} 
A distribution $t\in \mathcal{D}^\prime(U)$ belongs to $E_{s,N^*(I),p}(U)$ if $t\in E^\rho_{s,N^*(I),p}$
for some Euler vector field $\rho$. We define $E_{s,N^*(I)}(U)$ as 
the space of all distributions $t\in\mathcal{D}^\prime(U)$ such that
$t\in E_{s,N^*(I),p}(U)$ for all $p \in I\cap int(\overline{U})$. 
\end{defi}

It is immediate to deduce 
from Theorem \ref{thminvmuloc} and 
definition \ref{defEsmuloc} that
$E_{s,N^*(I)}$ satisfies the same restriction
and gluing properties as $E_{s,I}$.

We prove that
$E_{s,N^*(I)}(U)$ satisfies
a property of
diffeomorphism
invariance:
\begin{thm}\label{Esmicrolocdiffinvariance}
Let $I$ (resp $I^\prime$) be a closed embedded submanifold of $M$
(resp $M^\prime$), $U\subset M$ (resp $U^\prime\subset M^\prime$) open and $\Phi:U^\prime\mapsto U$ 
a diffeomorphism s.t. $\Phi(U^\prime\cap I^\prime)=I\cap U$.
Then $\Phi^*E_{s,N^*(I)}(U)=E_{s,N^*(I^\prime)}(U^\prime)$.
\end{thm}
\begin{proof}
By Theorem \ref{thminvmuloc}, we can localize the proof
at all points $p\in I\cap int(\overline{U})$. Let $p\in I\cap int(\overline{U})$,
then $t\in E_{s,N^*(I)}(U)$ implies
by definition 
that $t\in E^\rho_{s}(\mathcal{D}^\prime_\Gamma(V))$,
where $int(\overline{V})$ is a
neighborhood
of $p$, some Euler $\rho$
and $\Gamma$ satisfying the
conormal landing condition, 
which means that:
$$\lambda^{-s}e^{\log\lambda\rho*}t \text{ bounded }\mathcal{D}_\Gamma^\prime(V)$$
$$\Leftrightarrow \lambda^{-s}\Phi^* e^{\log\lambda\rho*}\Phi^{-1*} (\Phi^*t)
\text{ bounded in }\mathcal{D}_{\Phi^*\Gamma}^\prime(\Phi^{-1}(V))$$
because the pull--back by a diffeomorphism
is a bounded map from 
$\mathcal{D}_\Gamma^\prime(V)\mapsto \mathcal{D}_{\Phi^*\Gamma}^\prime(\Phi^{-1}(V))$,
$$\Leftrightarrow   \lambda^{-s} e^{\log\lambda(\Phi^{-1}_*\rho)*} (\Phi^*t)
\text{ bounded in } \mathcal{D}_{\Phi^*\Gamma}^\prime(\Phi^{-1}(V)) $$
where the vector field $\Phi^{-1}_*\rho$
is Euler by 
\ref{eulerdiffinvariance}.
Therefore $\Phi^*t$
is in $E_s^{\Phi^{-1}_*\rho}( \mathcal{D}_{\Phi^*\Gamma}^\prime(\Phi^{-1}(V)))$ at $p^\prime=\Phi^{-1}(p)$
where $\Phi^*\Gamma$ also 
satisfies the conormal landing condition
hence $\Phi^*t$ is locally in $E_{s,N^*(I),p}$ at $p$ and repeating the proof for all $p\in I\cap int(\overline{U})$
yields the claim.
\end{proof}

%
\section{The extension theorem
for $E_{s,N^*(I)}$.}
\label{extthmsectionmicroloc}
We are now ready to prove Theorem \ref{microlocextensionfinal}:
\begin{thm}
Let $U\subset M$ be some
open neighborhood of $I$. 
If $t\in E_{s,N^*(I)}(U\setminus I)$ 
then there exists an extension $\overline{t}$
with $WF(\overline{t})\subset WF(t)\cup N^*(I)$ and
$\overline{t}\in  E_{s^\prime,N^*(I)}(U)$, 
where $s^\prime = s$ if 
$s+d\notin -\mathbb{N}$ and $s^\prime< s$ otherwise.
\end{thm}
\begin{proof}
$t\in E_{s,N^*(I)}(U\setminus I)$ implies that
for all $p\in  I$, there is
some open chart $\psi:V_p\subset U\mapsto \mathbb{R}^{n+d}$, 
$\psi(I)\subset \mathbb{R}^n\times\{0\}$
where $\lambda^{-s}(\psi_*t)(x,\lambda h) $
is bounded in $\mathcal{D}_\Gamma^\prime(\psi\left(V_p\setminus I\right))$
for $\Gamma$ satisfying the conormal landing condition. Moreover, we must choose $V_p$ in such a way that
its image $U=\psi(V_p)\subset \mathbb{R}^{n+d}$
is of the form $U_1\times U_2$ where $U_1\subset\mathbb{R}^n, U_2\subset\mathbb{R}^d$
and $\lambda U_2\subset U_2,\forall\lambda\in[0,1]$.
 $\cup_{p\in I} V_p$ forms an open cover
of $I$, consider a locally finite 
subcover
$\cup_{a\in A} V_a$ and denote by $(\psi_a)_{a\in A}$
the corresponding charts.
For every $a\in A$, Theorem \ref{flatgenextensionthm} yields an extension
$\overline{\psi_{a*}t}$ of $\psi_{a*}t$ in $E_{s^\prime,N^*(I)}(\psi_a(V_a))$
and by diffeomorphism invariance of $E_{s^\prime,N^*(I)}$ (Theorem \ref{Esmicrolocdiffinvariance}), 
the element $\psi_a^*\overline{\psi_{a*}t}$
belongs to $E_{s^\prime,N^*(I)}(V_a)$. Choose a partition of unity
$(\varphi_a)_a$ subordinated to the open cover
$\cup_{a\in A} V_a$, then an extension 
of $t$ reads $\sum_{a\in A} \varphi_a  \psi_a^*\overline{\psi_{a*}t}+
\left(1-\sum_{a\in A} \varphi_a\right)t$ and belongs
to $E_{s^\prime,N^*(I)}(U)$ by the gluing property for $E_{s^\prime,N^*(I)}$.
\end{proof}
\section{Renormalized products.}
\label{renormprodsection}
In this section, we have a fixed 
Euler vector field $\rho$
and we scale only w.r.t. the flow generated
by $\rho$. 
We can now prove our 
Theorem~\ref{renormprodthm}
of renormalization 
of the product, we denote by 
$E^\rho_s(\mathcal{D}^\prime_\Gamma(U))$ the space of distributions
$t$ s.t. the family $\left(\lambda^{-s}e^{\log\lambda\rho*}t\right)_{\lambda\in(0,1]}$ is
bounded in $\mathcal{D}^\prime_\Gamma(U)$ for some $\rho$-convex set $U$ and some cone $\Gamma$ stable by scaling:
\begin{thm}
Let $\rho$ be some Euler vector field, $U$ some neighborhood of $I$, $\left(\Gamma_1,\Gamma_2\right)$ two cones in $T^\bullet \left(U\setminus I\right)$
which satisfy the conormal landing condition and $\Gamma_1\cap -\Gamma_2=\emptyset$.
Set $\Gamma=\left(\Gamma_1+\Gamma_2\right)\cup\Gamma_1\cup\Gamma_2$.
If $\Gamma_1+\Gamma_2$ satisfies the conormal landing condition then
there exists a bilinear map $\mathcal{R}$ satisfying the following properties: 
\begin{itemize}
\item $\mathcal{R}:(u_1,u_2)\in E^\rho_{s_1}\left(\mathcal{D}^\prime_{\Gamma_1}(U\setminus I) \right)\times E^\rho_{s_2}\left(\mathcal{D}^\prime_{\Gamma_2}(U\setminus I) \right) \mapsto \mathcal{R}(u_1u_2)\in E_{s,N^*(I)}\left(U\right), \forall s<s_1+s_2$
\item $\mathcal{R}(u_1u_2)=u_1u_2\text{ on }U\setminus I$
\item $\mathcal{R}(u_1u_2)\in  \mathcal{D}^\prime_{\Gamma\cup N^*(I)}(U).$
\end{itemize}
\end{thm}

\begin{proof}
The families $\left(\lambda^{-s_i}e^{\log\lambda\rho*}u_i\right)_{\lambda\in(0,1]}$ 
are bounded in $\mathcal{D}^\prime_{\Gamma_i}\left(U\setminus I\right)$.
By hypocontinuity of the H\"ormander product~\cite[Theorem 7.1]{Viet-wf2},
the family $\left(\lambda^{-s_1-s_2}e^{\log\lambda\rho*}(u_1u_2)\right)_{\lambda\in(0,1]}$ is bounded
in $\mathcal{D}^\prime_{\Gamma}(U\setminus I)$, 
$\Gamma$ still satisfies the
conormal landing condition by assumption then it follows by Theorem
\ref{microlocextensionfinal} that $u_1u_2$ admits an extension 
$\mathcal{R}(u_1u_2)$ in $E_{s,N^*(I)}\left(U\right)$ and $\mathcal{R}(u_1u_2)\in \mathcal{D}^\prime_{\Gamma\cup N^*(I)}(U)$.
\end{proof}

\section{Renormalization ambiguities.}
\label{renormambiguities} 
\subsection{Removable singularity theorems.}
First, we would like to start this section by a simple removable singularity theorem in the spirit of \cite[Theorems 5.2 and 6.1]{harvey1970removable}. 
In a renormalization procedure, there is always an ambiguity which is the ambiguity of the extension of the distribution. Indeed, two extensions always differ by a distribution supported on $I$. The removable singularity theorem states that if $s+d>0$ and if we demand that $t\in E_{s,I}(U\setminus I)$ should extend to $\overline{t}\in E_{s,I}(U)$, then the extension is \textbf{unique}.
Otherwise, if $-m-1 < s+d \leqslant -m $, then we bound the transversal order of the ambiguity.
We fix the coordinate system $(x^i,h^j)$ in $\mathbb{R}^{n+d}$ and $I=\{h=0\}$. The collection of coordinate functions $(h^j)_{1\leqslant j\leqslant d}$ defines a canonical collection of transverse vector fields $(\partial_{h^j})_j$.
We denote by $\delta_I$ the unique distribution such that 
\begin{eqnarray}
\forall\varphi\in \mathcal{D}(\mathbb{R}^{n+d}),\left\langle \delta_I,\varphi\right\rangle=\int_{\mathbb{R}^n} \varphi(x,0)d^nx.
\end{eqnarray}
If $t\in \mathcal{D}^\prime(\mathbb{R}^{n+d})$ with $\text{supp } t \subset I$, then
by \cite[Theorems 36,37 p.~101--102]{Schwartz-66} or \cite[Theorem 2.3.5]{HormanderI} there exist
unique distributions (once the system of transverse vector fields $\partial_{h^j}$ is fixed) $t_\alpha\in \mathcal{D}^\prime\left(\mathbb{R}^n\right)$, where each compact intersects
$\text{supp }t_\alpha$ for a finite number of multi--indices $\alpha$, such that $t(x,h)=\sum_\alpha t_\alpha(x) \partial_h^\alpha\delta_I(h)$
  or) where the $\partial_h^\alpha$ are derivatives in the \textbf{transverse} directions. 
\begin{thm}\label{removsing}
Let $t\in E_{s,I}(U\setminus I)$ and $\overline{t}\in E_{s^\prime,I}(U\setminus I)$
its extension given by Theorem (\ref{thm1}) and Theorem (\ref{thm2}) $s^\prime=s$ when $-s-d\notin\mathbb{N}$ or $\forall s^\prime<s$ otherwise.
Then $\tilde{t}$ is an extension in $E_{s^\prime,I}(U)$ 
if and only if 
\begin{equation}\label{Schwartzreps}
\tilde{t}(x,h)=\overline{t}(x,h)+\sum_{\alpha\leqslant m} t_\alpha(x) \partial_h^\alpha\delta_I(h),
\end{equation}
where $m$ is the integer part of $-s-d$. 
In particular when $s+d>0$ the extension is unique.
\end{thm}
Remark: when $-s-d$ is a nonnegative integer,
the counterterm is in
$E_{s,I}$ whereas the extension
is in $E_{s^\prime,I},\forall s^\prime<s$.

\begin{proof}
We scale an elementary distribution $\partial_h^\alpha\delta_I$:
$$\left\langle (\partial_h^\alpha\delta_I)_\lambda, \varphi\right\rangle=\lambda^{-d}\left\langle \partial_h^\alpha\delta_I, \varphi_{\lambda^{-1}}\right\rangle=(-1)^{\vert\alpha\vert}\lambda^{-d-\vert\alpha\vert} \left\langle \partial_h^\alpha\delta_I, \varphi\right\rangle $$
hence $\lambda^{-s}(\partial^\alpha\delta_I)_\lambda=\lambda^{-d-\vert\alpha\vert-s}\partial_h^\alpha\delta_I$ is bounded iff $d+s+\vert\alpha\vert\leqslant 0 \implies d+s\leqslant -\vert\alpha\vert $. 
When $s+d>0$, $\forall\alpha,\partial_h^\alpha\delta_I\notin E_{s,I}$ hence
any two extensions in $E_{s,I}(U)$ cannot differ
by a local counterterm of the form $\sum_\alpha t_\alpha \partial_h^\alpha\delta_I$.
When $-m-1<d+s\leqslant -m$ then $\lambda^{-s}(\partial_h^\alpha\delta_I)_\lambda$ is bounded iff
$s+d+\vert\alpha\vert\leqslant 0\Leftrightarrow -m\leqslant -\vert\alpha\vert\Leftrightarrow \vert\alpha\vert\leqslant m $. 
We deduce that $\partial_h^\alpha\delta_I \in E_{s,I}$ for all $\alpha\leqslant m$
which means that the scaling degree \textbf{bounds} the order $\vert\alpha\vert$ 
of the derivatives in the transverse directions. 
Assume there are two extensions in $E_{s,I}$, their difference is of the form $u=\sum_\alpha u_\alpha \partial_h^\alpha\delta_I$ by the structure theorem $(36)$ p.~101 in \cite{Schwartz-66} and is also in $E_{s,I}$
which means their difference equals $u=\sum_{\vert\alpha\vert\leqslant m} u_\alpha \partial_h^\alpha\delta_I$.
\end{proof}

\subsection{Counterterms on manifolds and conormal distributions.}

\paragraph{What happens in the case of manifolds ?}
From the point of view of 
L. Schwartz,
the only thing to keep in mind
is that a distribution
supported on a submanifold $I$ 
is always
well defined locally
and the representation 
of this distribution is 
unique once we fix
a system
of coordinate functions
$(h^j)_j$
which are transverse to $I$~\cite[Theorem 37 p.~102]{Schwartz-66}.
For any distribution $t_\alpha\in\mathcal{D}^\prime(I)$, if we denote by 
$i:I\hookrightarrow M$ the canonical embedding of $I$ in $M$ then 
$i_* t_\alpha$ is the push-forward of $t_\alpha$ in $M$:
$$\forall\varphi\in\mathcal{D}(M), \left\langle i_* t_\alpha, \varphi \right\rangle=\left\langle t_\alpha, \varphi\circ i \right\rangle.$$ 
The next lemma completes Theorem \ref{removsing}. 
Here the idea is that we add a constraint on the \textbf{local counterterm }$t$, namely that $WF(t)$ is contained in the conormal of $I$. Then we prove that the coefficients $t_\alpha$ appearing in the Schwartz representation
(\ref{Schwartzreps}) are in fact \textbf{smooth} functions.
\begin{lemm}\label{lemmbound}
Let $t\in\mathcal{D}^\prime(M)$ such that
$t$ is supported on $I$, then\\
1) $t$ has a unique 
decomposition as locally finite
linear combinations of transversal derivatives
of push-forward to $M$ of distributions $t_\alpha$ in $\mathcal{D}^\prime (I)$:
$t=\sum_{\alpha} \partial^\alpha_h\left(i_* t_\alpha\right)$,\\
and 2) $WF(t)$ is contained in the conormal of $I$ if and only if
$\forall \alpha$, $t_\alpha$ is smooth.
\end{lemm}
\begin{proof}
If $(t_\alpha)_\alpha$ are smooth then the wave front set
of the push--forward $i_*t_\alpha$ is contained
in the normal of the embedding $i$ denoted by $N_i$
\cite[2.3.1]{Viet-wf2}
which is nothing but the
conormal bundle $N^*(I)$ \cite[Example 2.9]{Viet-wf2}.
To prove the converse, 
in local coordinates,
let $$t(x,h)=\sum_\alpha \partial_h^\alpha\left(t_\alpha(x) \delta_I(h)\right)=\sum_\alpha t_\alpha(x) \partial_h^\alpha\delta_I(h).$$
Assume $t_\alpha$ is not smooth 
then $WF(t_\alpha)$ would be \textbf{non empty}. 
Then $WF(t_\alpha)$ 
contains an element $(x_0;\xi_0)$. 
Pick $\chi\in \mathcal{D}(R^n)$ such that $\chi(x_0)\neq 0$ 
then
$$\mathcal{F}(t_\alpha\chi \partial_h^\alpha\delta_I)(\xi,\eta)=\widehat{t_\alpha\chi}(\xi)(-i\eta)^\alpha, $$
hence we find a codirection $(\lambda \xi_0,\lambda \eta),\xi_0\neq 0$ in which the product  $\widehat{t_\alpha\chi}\widehat{\partial_h^\alpha\delta_I}$ is not rapidly decreasing, hence there is a point $(x,0)$ such that $(x,0;\xi_0,\eta_0)\in WF(t)$~\cite[Lemma 8.2.1]{HormanderI} 
which is in contradiction with the fact that $WF(t)\subset N^*(I)=\{(x,0,0,\eta)|\eta\neq 0\}$.
\end{proof}

Combining with Theorem \ref{removsing}, 
we obtain:
\begin{coro}
Let $t\in\mathcal{D}^\prime(\mathbb{R}^{n+d})$ and $\text{supp }t\subset I$.
If $WF(t)\subset N^*(I)$ and $t\in E_{s,N^*(I)}(\mathbb{R}^{n+d}), -m-1<s+d\leqslant -m$, then $t(x,h)=\sum_\alpha t_\alpha(x) \partial_h^\alpha\delta_I(h)$, where $\forall\alpha$, $t_\alpha\in C^\infty\left(\mathbb{R}^n\right)$ and $\vert\alpha\vert\leqslant m$.
\end{coro}
\begin{coro}
Let $M$ be a smooth manifold
and $I$ a closed embedded submanifold.
For $-m-1< s+d\leqslant-m $, the space of distributions $t\in E_{s,N^*(I)}(M)$ 
such that $\text{supp }t\in I$ and $WF(t)$ is contained in the conormal of $I$ is a finitely generated module of \textbf{rank} $\frac{m+d!}{m!d!}$ over the ring $C^\infty(I)$. 
\end{coro}
\begin{proof}
In each local chart $(x,h)$ where $I=\{h=0\}$,
$t=\sum_\alpha t_\alpha(x) \partial_h^\alpha\delta_I(h)$ where the lenght $\vert\alpha\vert$ is bounded by $m$ by the above corollary and $\forall\alpha$, $t_\alpha\in C^\infty\left(I\right)$. This improves on the result given 
by the structure theorem of Laurent Schwartz 
since we now know that the $t_\alpha$ are smooth.
\end{proof}

\section{Appendix.}
\paragraph{Wave front set of the kernels of the operators $I_m,R_\alpha$.}
In this
part, we calculate
the wave front set of the kernels of the operators $I_m,R_\alpha$
introduced in the
proof of Theorem \ref{thm3}.
Recall $I=\mathbb{R}^n\times\{0\}$ 
is the vector subspace
$\{h=0\}$, we define 
the projection
$\pi:(x,h)\in\mathbb{R}^{n+d}
\longmapsto 
(x,0)\in\mathbb{R}^n\times\{0\}$,
the inclusion 
$i:\mathbb{R}^n\times \{0\} \hookrightarrow \mathbb{R}^{n+d}$, 
the operator $I_m$ of
projection
on the 
Taylor remainder of degree $m$:
\begin{eqnarray*}
I_m:=\varphi\in C^\infty(\mathbb{R}^{n+d})\longmapsto I_m\varphi= \varphi-P_m\varphi\in C^\infty(\mathbb{R}^{n+d})\\
P_m\varphi=\sum_{\vert\alpha\vert\leqslant m} \frac{h^\alpha}{\alpha!} \pi^*\left(i^*\partial_h^\alpha \varphi\right)\\
I_m\varphi=\frac{1}{m!}\sum_{\vert\alpha\vert=m+1}h^\alpha\int_0^1(1-t)^m \left(\partial_h^\alpha \varphi\right)_tdt.
\end{eqnarray*}
We also introduce the operators
$(R_\alpha)_{\{\vert\alpha\vert=m+1\}}$:
\begin{eqnarray*}
I_m
=\sum_{\vert\alpha\vert=m+1}h^\alpha R_\alpha.
\end{eqnarray*}

We next explain how to calculate the Schwartz kernels
of $I_m,R_\alpha$ which are distributions in 
$\mathcal{D}^\prime(\mathbb{R}^{n+d}\times\mathbb{R}^{n+d})$
and their wave front set.
We double
the space $\mathbb{R}^{n+d}$
and we work in $\mathbb{R}^{n+d}\times\mathbb{R}^{n+d}$
with coordinates $(x,h,x^\prime,h^\prime)$. We denote
by $\delta\in\mathcal{D}^\prime(\mathbb{R}^{n+d})$ the delta distribution
supported at $(0,0)\in\mathbb{R}^{n+d}$ 
and 
$\delta_\Delta(.,.)\in\mathcal{D}^\prime(\mathbb{R}^{n+d}\times \mathbb{R}^{n+d})$ 
the 
delta distribution supported
on the diagonal $\Delta$ 
in $\mathbb{R}^{n+d}\times \mathbb{R}^{n+d}$
where we have the relation
$\delta_\Delta((x,h),(x^\prime,h^\prime))=\delta(x-x^\prime,h-h^\prime)$. 
The Schwartz 
kernel of $I_m$
is the distribution
defined as:
\begin{eqnarray}
I_m(.,.)&=&\delta_\Delta(.,.) -\sum_{\vert\alpha\vert\leqslant m} \frac{h^\alpha}{\alpha!} \pi^*\left(i^*\partial_h^\alpha \delta_\Delta(.,.) \right)\\
&=&\frac{1}{m!}\sum_{\vert\alpha\vert=m+1}h^\alpha\int_0^1(1-t)^m \partial_h^\alpha  \delta_\Delta(\Phi(t,.),.)dt,
\end{eqnarray}
where $\Phi(t,x,h)=(x,th) .$
We also need to define
Schwartz kernels $R_\alpha$:
\begin{eqnarray*}
R_\alpha(.,.)&=&\frac{1}{m!}\int_0^1(1-t)^m \partial_h^\alpha  \delta(\Phi(t,.),.)dt\\
\text{ where }
I_m(.,.)
&=&\sum_{\vert\alpha\vert=m+1}h^\alpha R_\alpha(.,.).
\end{eqnarray*} 

\begin{lemm}\label{WFI_m}
Let $I_m(.,.)$ and $R_\alpha(.,.)$ be defined
as above then
\begin{eqnarray}
WF\left( R_\alpha(.,.)\right)\subset  \{(x,h,x,th;\xi,t\eta,-\xi,-\eta) |t\in[0,1], (\xi,\eta)\neq (0,0)\}.
\end{eqnarray}
and
\begin{equation}\label{wfim}
WF(I_m(.,.))\subset  \{(x,h,x,th;\xi,t\eta,-\xi,-\eta) |t\in[0,1], (\xi,\eta)\neq (0,0)\}.
\end{equation}
\end{lemm}
\begin{proof}
Let us calculate
$WFI_m(.,.)$, the idea is 
to work in ``extended phase space'' $[0,1]\times \mathbb{R}^{n+d}\times \mathbb{R}^{n+d}$
with coordinates
$(t,x,h,x^\prime,h^\prime)$.
Consider the map 
\begin{eqnarray*}
\Phi:=(t,x,h,x^\prime,h^\prime)\in[0,1]\times \mathbb{R}^{n+d}\times \mathbb{R}^{n+d}\longmapsto (x,th,x^\prime,h^\prime)\in \mathbb{R}^{n+d}\times \mathbb{R}^{n+d},
\end{eqnarray*}
then $(\Phi^* \delta)(t,x,h,x^\prime,h^\prime)=\delta((x,th),(x^\prime,h^\prime))$ 
and application of the pull--back theorem
\cite[Proposition 6.1]{Viet-wf2} implies that
\begin{eqnarray}\label{pullbackdelta}
WF(\Phi^* \partial_h^\alpha\delta(.,.)))\subset
\{(t,x,h,x^\prime,h^\prime;\tau,\xi,t\eta,-\xi,-\eta) | (x,th)=(x^\prime,h^\prime) 
\text{ and } \tau=\left\langle h.\eta\right\rangle, (\xi,\eta)\neq (0,0) \}.
\end{eqnarray} 
We also note that $m!R_\alpha$
is just the integral of 
$f=1_{[0,1]}(1-t)^m\Phi^* \partial_h^\alpha\delta(.,.))dt$
over $[0,1]$, in other words, 
it is the push--forward of
$f$ 
by the projection $\mathbf{p}:\mathbb{R}\times \mathbb{R}^{2(n+d)}\mapsto \mathbb{R}^{2(n+d)}$. From the bound (\ref{pullbackdelta}) on 
$WF(\Phi^* \partial_h^\alpha\delta(.,.)))$ and the
behaviour of wave front sets under product, we find the 
rough upper bound:
\begin{eqnarray*}
WF(f)
\subset \Xi=\{(t,x,h,x^\prime,h^\prime;\tau,\xi,t\eta,-\xi,-\eta) | (x,th)=(x^\prime,h^\prime) 
,(\xi,\eta)\neq (0,0)\}.
\end{eqnarray*}
Finally, from the relation
$R_\alpha=\frac{\mathbf{p}_*f}{m!}, f=1_{[0,1]}(1-t)^m\Phi^* \partial_h^\alpha\delta(.,.))dt$
we find that
\begin{eqnarray*}
WF R_\alpha(.,.)&\subset& \mathbf{p}_*WF(f)\subset \mathbf{p}_*\Xi \\
\implies WF R_\alpha(.,.)&\subset & \{(x,h,x,th;\xi,t\eta,-\xi,-\eta) |t\in[0,1], (\xi,\eta)\neq (0,0)\}.
\end{eqnarray*}
\end{proof}

We also need the wave front set
of the Schwartz kernel of the operator
$\varphi\longmapsto P_m\varphi$
which projects $\varphi$ on its 
``Taylor polynomial'':
\begin{eqnarray}
\forall \vert\alpha\vert\leqslant m, 
WF\left(\pi^*(i^*\partial_h^\alpha\delta_\Delta(.,.))\right)\subset  \{(x,h,x,0;\xi,0,-\xi,-\eta) | (\xi,\eta)\neq (0,0)\}.
\end{eqnarray}
Note the important fact 
that $WF\left(\pi^*(i^*\partial_h^\alpha\delta_\Delta(.,.))\right)\subset WF(I_m(.,.))$.
\begin{lemm}\label{wfPm}
Let $\delta_\Delta((x,h),(x^\prime,h^\prime))$ be the delta function
of the diagonal $\Delta\subset \mathbb{R}^{n+d}\times\mathbb{R}^{n+d}$,
$i:x\longmapsto (x,0)$ 
the inclusion of $\mathbb{R}^n$ in $\mathbb{R}^{n+d}$ and $\pi$
the projection $(x,h)\in\mathbb{R}^{n+d}\longmapsto x\in\mathbb{R}^n$. 
The Schwartz kernel
of the linear map $\varphi\longmapsto \pi^*(i^*\partial_h^\alpha\varphi)$
is $\pi^*(i^*\partial_h^\alpha\delta_\Delta)$,
\begin{equation}
WF(\pi^*(i^*\partial_h^\alpha\delta_\Delta))\subset\{(x,h,x,0;\xi,0,-\xi,-\eta) 
|  (\xi,\eta)\neq (0,0)\}.
\end{equation}
\end{lemm}
\begin{proof}
First, we have:
$WF (i^*\partial_h^\alpha\delta_\Delta)\subset
 \{(x,x,0;\xi,-\xi,-\eta),\,\ (\xi,\eta)\neq (0,0)\}$, then
$$WF(\pi^*(i^*\partial_h^\alpha\delta_\Delta))
\subset
\{(x,h,x^\prime,h^\prime ; \xi,0,\xi^\prime,\eta^\prime)| 
(x,x^\prime,h^\prime;\xi,\xi^\prime,\eta^\prime)\in 
WF (i^*\partial_h^\alpha\delta_\Delta) \} $$
$$=\{(x,h,x,0; \xi,0,-\xi,-\eta)| (\xi,\eta)\neq (0,0)\}. $$

\end{proof}

\paragraph{Technical Lemma.}

In this part, we prove the main technical Lemma
which is essential in the proof of the main Theorems
of section \ref{flatspacesection} and we follow
its terminology and notations.
\begin{lemm}\label{propfamily2}
Let $U\subset\mathbb{R}^{n+d}$ be a convex set, for $\varepsilon\geqslant 0$, 
$1_{[\varepsilon,1]}$ is the indicator
function of $[\varepsilon,1]$. Set
\begin{eqnarray}
V=\{
\left(\begin{array}{ccc}
\lambda & ;&\widehat{\lambda} \\
x &;& \widehat{\xi}\\
h &;&\widehat{\eta}  
\end{array}\right)
| \left(\begin{array}{ccc}
x &;& \widehat{\xi}\\
h &;& \widehat{\eta}  
\end{array}\right)\in  \Gamma\cup \underline{0}, (x,h)\in\text{supp }\psi\}.
\end{eqnarray}
Let $B$ be some
bounded subset in $E_{s}(\mathcal{D}^\prime_\Gamma(U))$.
For all function $f\in L^{1}([0,1])\cap C^\infty(0,1)$, for all $t\in B$, the family $\left(f1_{[\varepsilon,1]}\lambda^{-s}\Phi^*t\right)_{\varepsilon\in[0,1],t\in B}$ 
is bounded in $\mathcal{D}^\prime_V(\mathbb{R}\times U)$.
\end{lemm}
\begin{proof}
We first prove that
$\left(f1_{[\varepsilon,1]}\lambda^{-s}\Phi^*t\right)_{\varepsilon\in[0,1]}$  
is weakly bounded in $\mathcal{D}^\prime(\mathbb{R}\times U)$.
$\lambda^{-s}t_\lambda$ is bounded in $\mathcal{D}^\prime(U)$
therefore by the uniform boundedness principle in Fr\'echet space~\cite{Meise},
\begin{eqnarray*}
\forall K\subset U \text{ compact },\exists m\in\mathbb{N},\exists C>0, \forall \varphi\in\mathcal{D}_K(U),  
\sup_{\lambda\in[0,1]} 
\vert\left\langle \lambda^{-s}t_\lambda,\varphi\right\rangle \vert\leqslant C \pi_{m,K}(\varphi).
\end{eqnarray*}
If $t$ is in a bounded subset $B$ of $E_s(\mathcal{D}^\prime_\Gamma(U))$, then
one can choose the constant $C$ independent of $t\in B$.
It follows easily that for all subset of the form $\left(\mathbb{R}\times K\right)\subset 
\left(\mathbb{R}\times U\right)$:
\begin{eqnarray*}
\exists m\in\mathbb{N}, C \geqslant 0, \text{ such that } 
\forall \varphi\in\mathcal{D}_{\mathbb{R}\times K}(\mathbb{R}\times U), \forall\varepsilon\geqslant 0,
&&\\
\vert\int_{[\varepsilon,1]\times\mathbb{R}^{n+d}}f(\lambda)\lambda^{-s}t(x,\lambda h)\varphi(\lambda,x,h)d\lambda dxdh\vert 
&\leqslant &\Vert f\Vert_{L^1([0,1])}\sup_{\lambda\in[0,1]} 
\vert\left\langle \lambda^{-s}t_\lambda,\varphi(\lambda,.)\right\rangle \vert \\
&\leqslant & C \Vert f\Vert_{L^1([0,1])}\sup_{\lambda\in[0,1]}\pi_{m,K}(\varphi(\lambda,.))  \\
&\leqslant & C \Vert f\Vert_{L^1([0,1])}\pi_{m,[0,1]\times K}(\varphi).
\end{eqnarray*}
 

For all $(\lambda,x,h;\tau,\xi,\eta)\notin V$,
there is a conic set $W\subset \mathbb{R}^{n+d}\setminus \{0\}$, 
a test function
$\varphi_2\in \mathcal{D}(U)$ such that
$(x,h;\xi,\eta)\in\text{supp }\varphi_2\times W$ 
and $(\text{supp }\varphi_2\times W)\cap \Gamma=\emptyset$.
Let $\varphi(\lambda,x,h)=\varphi_1(\lambda)\varphi_2(x,h)$
for some $\varphi_1,\varphi_1(\lambda)\neq 0$ in
$\mathcal{D}(\mathbb{R})$
and we define a conic neighborhood $W^\prime$ 
of $(\tau_0,\xi_0,\eta_0)$ 
as follows
$W^\prime=\{(\tau,\xi,\eta) | 
\vert \tau\vert\leqslant 
2\frac{\vert\tau_0\vert}{\vert\xi_0\vert+\vert\eta_0\vert}
(\vert\xi\vert+\vert\eta\vert), (\xi,\eta)\in W \}$.
We find that $\forall (\tau,\xi,\eta)\in W^\prime$:
\begin{eqnarray*}
\vert\int_\varepsilon^1 d\lambda f(\lambda)\left\langle \lambda^{-s}t_\lambda,\varphi_2 e^{i(x.\xi+h.\eta)}\right\rangle \varphi_1(\lambda) e^{i\lambda.\tau} \vert
 &=&\vert\int_\varepsilon^1 d\lambda\widehat{\lambda^{-s}t_\lambda\varphi_2}(\xi,\eta)f(\lambda) \varphi_1(\lambda) e^{i\lambda.\tau} \vert\\
&\leqslant& \Vert \varphi_1\Vert_{L^\infty(\mathbb{R})}\Vert f\Vert_{L^1[0,1]}
\Vert \lambda^{-s}t_\lambda \Vert_{N,W,\varphi_2}(1+ \vert\xi\vert+\vert\eta\vert)^{-N}\\
&\leqslant&
 C \Vert \varphi_1\Vert_{L^\infty(\mathbb{R})}\Vert f\Vert_{L^1[0,1]} 
 \Vert \lambda^{-s}t_\lambda \Vert_{N,W,\varphi_2}(1+\vert\tau\vert+\vert\xi\vert+\vert\eta\vert)^{-N}
\end{eqnarray*}
where $C=(1+
2\frac{\vert\tau_0\vert}{\vert\xi_0\vert+\vert\eta_0\vert})^{N}$. 
Therefore, $\forall (\lambda,x,h;\tau,\xi,\eta)\notin \Lambda,\exists \chi\in\mathcal{D}(\mathbb{R}\times U)$
and a closed conic set $W^\prime$ such that
$\chi(\lambda,x,h)\neq 0$, $\left(\text{supp }\chi\times W^\prime\right)\cap \Lambda=\emptyset$ and
the following estimate is satisfied:
\begin{equation}
\forall N,\exists C, \Vert f\lambda^{-s}\Phi^* t\Vert_{N,W^\prime,\chi}\leqslant
 C
\sup_{\lambda\in[0,1]} \Vert \lambda^{-s}t_\lambda \Vert_{N,W,\varphi}
\end{equation}
for some continuous seminorm $\sup_{\lambda\in[0,1]} \Vert \lambda^{-s}t_\lambda \Vert_{N,W,\varphi}$ of 
$E_s(\mathcal{D}^\prime_\Gamma(U))$ 
and where the constant
$C$ does not depend on $t$. 

 It follows easily from the above 
that the family $\left( f1_{[\varepsilon,1]}\lambda^{-s}\Phi^*t \right)_{\varepsilon\in (0,1]}$
is bounded in $\mathcal{D}^\prime_V(\mathbb{R}\times U)$.
\end{proof}

\subsubsection{The symplectic geometry of the vector fields tangent to $I$ and of the diffeomorphisms leaving $I$ invariant.}
We will work at the infinitesimal level 
within the class $\mathfrak{g}$ 
of vector fields tangent to $I$ 
defined by H\"ormander~\cite[Lemma (18.2.5)]{HormanderIII}. First recall 
their definition
in coordinates $(x,h)$ where $I=\{h=0\}$:  
the vector fields $X$ tangent to $I$ are of the form
$$ h^ja_j^i(x,h)\partial_{h^i} + b^i(x,h)\partial_{x^i} $$
and they form an infinite dimensional Lie algebra denoted by $\mathfrak{g}$ which is a Lie subalgebra of $Vect(M)$. 
Actually, these vector fields form a module over the ring $C^\infty(M)$ finitely generated by the vector fields $h^i\partial_{h^j},\partial_{x^i}$. 
This module is naturally filtered by the vanishing order of the vector field on $I$.
\begin{defi}
Let $\mathcal{I}$ be the ideal of functions vanishing on $I$. For $k\in\mathbb{N}$,
let $F_k$ be the submodule of vector fields tangent to $I$ defined as follows, $X\in F_k$ if 
$X\mathcal{I}\subset\mathcal{I}^{k+1}$.
\end{defi}
This definition of the filtration is completely coordinate invariant. We also immediately have $F_{k+1}\subset F_k$. Note that $F_0=\mathfrak{g}$.
\paragraph{Cotangent lift of vector fields.} 
 We recall the following fact, 
any vector field $X\in Vect(M)$
lifts functorially to a
\emph{Hamiltonian vector field} $X^*\in Vect(T^* M)$ by the following procedure  
which is beautifully described in \cite[p.~34]{chriss2010representation}: 
$$X=a^i\frac{\partial}{\partial x^i}+b^j\frac{\partial}{\partial h^j}\in Vect(M) \overset{\sigma}{\mapsto} \sigma(X)=a^i\xi_i+b^j\eta_j\in C^\infty(T^* M)$$ 
$$ \mapsto X^*=\left\lbrace \sigma(X),.  \right\rbrace= a^i\frac{\partial}{\partial x^i}+
b^j\frac{\partial}{\partial h^j}-\frac{\partial (a^i\xi_i+b^j\eta_j)}{\partial x^i}\frac{\partial}{\partial \xi_i}-
\frac{\partial (a^i\xi_i+b^j\eta_j)}{\partial h^j}\frac{\partial}{\partial \eta_j},$$
where $\left\lbrace .,.  \right\rbrace $ is the Poisson bracket of $T^* M$. 

\begin{lemm}\label{lemmlift}
Let $X$ be a vector field in $\mathfrak{g}$.
If $X\in F_1$, then $X^*$ vanishes on the conormal $N^*(I)$ of $I$ and $N^*(I)$ is contained in the set of fixed points of the symplectomorphism $e^{X^*}$.
\end{lemm}
\begin{proof} 
If $X\in F_1$, then $\sigma(X)= h^jh^ia^l_{ji}(x,h)\eta_l + h^ib_{i}^l(x,h)\xi_l $
where $a^l_{ji},b_{i}^l$ are smooth functions on $T^*M$ 
by the Hadamard lemma. 
The symplectic gradient $X^*$ is given by the formula
$$X^*= \frac{\partial \sigma(X)}{\partial \xi_i}\partial_{x^i} - \frac{\partial \sigma(X)}{\partial x^i}\partial_{\xi_i}+\frac{\partial \sigma(X)}{\partial \eta_i}\partial_{h^i} - \frac{\partial \sigma(X)}{\partial h^i}\partial_{\eta_i},$$ 
thus $X^*=0$ when $\xi=0,h=0$ which means $X^*=0$ on the conormal $N^*(I)$.
\end{proof}

\begin{lemm}\label{Idlift}
Let $\rho_1,\rho_2$ be two 
Euler vector fields and $\Phi(\lambda)=e^{-\log\lambda \rho_1}\circ e^{\log\lambda \rho_2}$.
Then the cotangent lift $T^* \Phi(\lambda)$
restricted to $N^*(I)$ is the identity map:
$$T^* \Phi(\lambda)|_{N^*(I)}=Id|_{N^*(I)}.$$
\end{lemm}
\begin{proof}
Let us set
\begin{equation}\label{conjugchapt1}
\Phi(\lambda)=e^{-\log\lambda \rho_1}\circ e^{\log\lambda \rho_2} 
\end{equation}
which is a family of diffeomorphisms
which depends smoothly in $\lambda\in[0,1]$
according to \ref{propositionvariablefamily}.
The proof is similar to the proof of proposition \ref{propositionvariablefamily},
$\Phi(\lambda)$ satisfies the differential equation:
\begin{equation}
\lambda\frac{d\Phi(\lambda)}{d\lambda}=e^{-\log\lambda\rho_1}\left(\rho_2-\rho_1\right)e^{\log\lambda\rho_1}\Phi(\lambda) \text{ where } \Phi(1)=Id
\end{equation}
we reformulated this differential equation as
\begin{equation}
\frac{d\Phi(\lambda)}{d\lambda}=X(\lambda)\Phi(\lambda),\Phi(1)=Id
\end{equation}
where the vector field $X(\lambda)=\frac{1}{\lambda}e^{-\log\lambda\rho_1}\left(\rho_2-\rho_1\right)e^{\log\lambda\rho_1}$ depends smoothly in $\lambda\in [0,1]$.
The cotangent lift $T^* \Phi_\lambda$ satisfies the differential equation
\begin{equation}
\frac{dT^*\Phi(\lambda)}{d\lambda}=X^*(\lambda)T^*\Phi(\lambda),T^* \Phi(1)=Id
\end{equation} 
Notice that $\forall \lambda\in[0,1], X(\lambda)\in F_1$ which implies that for all $\lambda$ the lifted Hamiltonian vector field $X^*(\lambda)$ will vanish on $N^*(I)$ by the lemma (\ref{lemmlift}). Since $T^* \Phi(1)=Id$ obviously fixes the conormal, this immediately implies that $\forall\lambda, T^*\Phi(\lambda)|_{N^*(I)}=Id|_{N^*(I)}$.
\end{proof}

\end{document}